\documentclass[11pt]{article}
 \usepackage[top=0.75in, bottom=0.75in, left=1in, right=1in]{geometry}
 \usepackage{amsmath}
 \usepackage{amssymb}
 \usepackage{amsthm}
 \usepackage{amsmath}
 \usepackage{setspace}
 \usepackage{xcolor}
 \usepackage{fancyhdr}
 \usepackage{amssymb}
 \usepackage{amsthm}
     \usepackage{cite}
     \usepackage{tabularx}
     \usepackage{graphicx}
     \usepackage{epstopdf}
     \usepackage{epsfig}
     \usepackage{float}
     \usepackage{ listings}
    \usepackage{amsmath}
\usepackage{amsthm}
\usepackage{setspace}
\usepackage{xcolor}
\usepackage{fancyhdr}
\usepackage{amssymb}
\usepackage{listings}
 \usepackage{cite}
 \usepackage{url}
  \usepackage{tabularx}
  \usepackage{graphicx}
  \usepackage{epstopdf}
  \usepackage{epsfig}
  \usepackage{float}
  \usepackage{ listings}
  \usepackage{algorithm}
\usepackage{algorithmic}

 \theoremstyle{plain}
 \newtheorem{thm}{Theorem}[section]
 \newtheorem{prop}[thm]{Proposition}
 \newtheorem{lem}[thm]{Lemma}
 
 \theoremstyle{definition}

 \theoremstyle{remark}
 \newtheorem{remark}[thm]{Remark}

 \let\b=\beta

 \let\e=\varepsilon

 \let\lam=\lambda

 \let\f=\frac

 \let\wt=\widetilde
 \let\td = \tilde
 \let\wh=\widehat

 \let\teq \triangleq

 \def\cF{{\cal F}}

\def\lf{\lfloor}
\def\rf{\rfloor}
 \def\lt{\left}
 \def\rt{\right}

 \def\th{\theta}
 \def\Cov{\mathrm{Cov}}
\def\diag {\mathrm{diag}}
\def\sign{\mathrm{sign}}
\def\Var{\mathrm{Var}}

 \newcommand{\beq}{\begin{equation}}
 \newcommand{\eeq}{\end{equation}}
  \newcommand{\bal}{\begin{aligned} }
  \newcommand{\eal}{\end{aligned}}
 \newcommand{\ben}{\begin{eqnarray}}
 \newcommand{\een}{\end{eqnarray}}
 \newcommand{\beno}{\begin{eqnarray*}}
 \newcommand{\eeno}{\end{eqnarray*}}

 \newcommand{\one}{\mathbf{1}}

 \newcommand{\BB}{\mathbf{B}}
 \newcommand{\CC}{\mathbf{C}}
 
  \newcommand{\EE}{\mathbf{E}}

 \newcommand{\II}{\mathbf{I}}
 
 \newcommand{\MM}{\mathbf{M}}

 \newcommand{\PP}{\mathbf{P}}
 
 \newcommand{\bS}{\mathbf{S}}

 \newcommand{\UU}{\mathbf{U}}

 \newcommand{\tr}{\mathrm{Tr}}

    \DeclareMathOperator*{\argmin}{argmin}
    
    \date{\today}
    \author{Jiajie Chen \footnote{ Applied and Computational Mathematics, Caltech; Email: jchen@caltech.edu}, Anthony Hou \footnote{Department of Statistics, Harvard University; Email: ahou@college.harvard.edu}, Thomas Y. Hou \footnote{Applied and Computational Mathematics, Caltech; Email: hou@cms.caltech.edu}}
    \title{A Pseudo Knockoff Filter for Correlated Features}
 
\begin{document}
 \maketitle{}
 \begin{abstract}
 In \cite{Candes}, the authors introduced a new variable selection procedure called the knockoff filter to control the false discovery rate (FDR) and proved that this method achieves exact FDR control. Inspired by the work of \cite{Candes}, we propose  a pseudo-knockoff filter that inherits some advantages of the original knockoff filter and has more flexibility in constructing its knockoff matrix. Moreover, we perform a number of numerical experiments that seem to suggest that the pseudo knockoff filter with the half Lasso statistic has FDR control and offers more power than the original knockoff filter with the Lasso Path or the half Lasso statistic for the numerical examples that we consider in this paper. Although we cannot establish rigorous FDR control for the pseudo knockoff filter, we provide some partial analysis of the pseudo knockoff filter with the half Lasso statistic and establish a uniform FDP bound and
an expectation inequality.
 \end{abstract}
 \section{Introduction}

 In many applications, we need to study a statistical model that consists of a response variable and a large number of potential explanatory variables and determine  which variables are truly associated with the response. In \cite{Candes}, Barber and Cand\`es introduce the knockoff filter to control the FDR in a statistical linear model. More specifically, the knockoff filter constructs knockoff variables that mimic the correlation structure of the true feature variables to obtain exact FDR control in finite sample settings. It has been demonstrated that this method has more power than existing selection rules when the proportion of null variables is high.

 \subsection{A brief review of the knockoff filter}
 Consider the following linear regression model $ y=X\beta+ \epsilon$ where the feature matrix $X$ is a $n\times p$ ($n \geq 2 p$) matrix with full rank, its columns have been normalized to be unit vectors in the $l^2$ norm, and $\epsilon$ is a Gaussian noise $N(0,\sigma^2I_n)$. The knockoff filter begins with the construction of a knockoff matrix $\tilde{X}$ that obeys
\begin{equation}\label{revko1}
   \tilde{X}^T\tilde{X} = X^TX, \quad \tilde{X}^T X = X^T X - \diag(s) ,
 \end{equation}
   where $s_i \in [0,1],i=1,2,...,p$. The positive definiteness of the Gram matrix $[X \tilde{X}]^T[X \tilde{X}]$ requires
   \begin{equation}\label{revko2}
   \diag(s) \preceq 2X^TX.
   \end{equation}
 The first condition in (\ref{revko1}) ensures that $\tilde{X}$ has the same covariance structure as the original feature matrix $X$. The second condition in (\ref{revko1}) guarantees that the correlations between distinct original and knockoff variables are the same as those between the original variables.
The power (the expected proportion of true discoveries) of the knockoff filter depends critically on the value of $s_i$. A general guideline in constructing the knockoff matrix is to choose $s_j$ as large as possible to maximize the difference between $X_j$ and its knockoff $\tilde{X}_j$. Next, we choose a statistic, $W_j$, for each pair $X_j,\tilde{X}_j$ by using the Gram matrix $[X \ \tilde{X}]^T [X \tilde{X}]$ and the marginal correlation $[X \ \tilde{X}]^T y$. In addition, $W_j$ satisfies a flip-coin property that swapping arbitrary pair $X_j, \td{X}_j$ only changes the sign of $W_j$ but keeps the sign of other $W_i$ ($i \neq j$)  unchanged. The construction of the knockoff features and the symmetry of the test statistic ensure that the signs of the $W_j$'s are i.i.d. random for the ``null hypotheses''. This property plays a crucial role in obtaining exact FDR control by using a supermartingale argument.

One of the knockoff statistics considered in \cite{Candes} is the Lasso path statistic, which is defined as $W_j = \max(Z_j , \tilde{Z}_j) \cdot \sign(Z_j - \tilde{Z}_j)$, where $Z_j$ and $\tilde{Z}_j$ are the solutions of the Lasso path problem given below:
 \[
   \begin{aligned}
        &(\hat{\beta}(\lambda),\tilde{\beta}(\lambda)) =\argmin_{ (b,\tilde{b}) } \left\{ \frac{1}{2} ||y - Xb -\tilde{X}\tilde{b}  ||_2^2 + \lambda (|| b||_1 + ||\tilde{b} ||_1) \right\} \; , \\
        & Z_j = \sup\{ \lambda : \hat{\beta}_j(\lambda) \ne 0 \},\  \tilde{Z}_j = \sup\{ \lambda:  \tilde{\beta}_j(\lambda) \ne 0 \}.\\
   \end{aligned}
   \]

 If $X_j$ is a nonnull, it has a non-trivial effect on $y$ and should enter the model earlier than its knockoff $\td{X}_j$, resulting in a positive $W_j$.
A large positive $W_j$ implies that there is a high probability that the variable $j$ is a nonnull. This consideration suggests that we select the variable $j$ with positive $W_j$ larger than a data-dependent threshold $T$, $\hat{S} \triangleq \{j: W_j \geq T  \}$, where $T$ is defined below
 \begin{equation}\label{def-T0}
   \begin{aligned}
   &T \triangleq   \min \left\{ t>0 :        \frac{ 1+  \#\{   j: W_j   \leq -t              \}  }{      \#\{   j: W_j   \geq t\}   \vee 1   }  \leq q  \right\}  \; . %\quad \hat{S} \triangleq \{j: W_j \geq T  \} \; . \\
  \end{aligned}
 \end{equation}
 The false discovery proportion (FDP) of the knockoff  filter and its estimate at threshold $t$ are given by
% at any threshold $t>0$ is defined as
 \beq\label{def:FDP}
 FDP(t) \teq   \frac{  \#\{   j : W_j  \geq t  \  \& \  \b_j = 0          \}  }{      \#\{   j: W_j   \geq t\}   \vee 1   }, \quad \widehat{FDP}(t) \teq       \frac{ 1+  \#\{   j: W_j   \leq -t              \}  }{      \#\{   j: W_j   \geq t\}   \vee 1   } .
 \eeq

 The FDR is the expectation of FDP. The i.i.d signs for the null $W_j$ enables one to construct a supermartingale $M_t$ with respect to an appropriate backward filtration $\cF_t$ such that
 \beq\label{eq:mart}
 \f{ FDP(t)} { \wh{ FDP(t)}} \leq  \f{     \#\{   j : W_j  \geq t  \  \& \  \b_j = 0\}        }{    1+  \#\{   j: W_j   \leq -t  \ \& \ \b_j = 0             \}   } \triangleq  M_t,  \quad  E [M_t] \leq 1.
 \eeq
 The threshold $T$ defined in \eqref{def-T0} gives a stopping time. Using the definition of $T$ and the stopping time theorem, the authors in \cite{Candes} obtained
 $ E\lt[ FDP(T) / q   \rt]  \leq  E\lt[    FDP(T) /  \wh{ FDP(T)} \rt] \leq E[ M_T]  \leq 1 $.
 The main result in \cite{Candes} is that the knockoff procedure controls the FDR
 \[
  FDR \triangleq  E[ FDP(T)   ] \leq q \; .
  \]

In a subsequent paper \cite{Candes2}, Barber and Cand\`es developed a framework for high-dimensional linear model with $p \ge n$. The knockoff filter has been further generalized to the model-free framework in \cite{Candes3}. The model-free knockoffs provide valid inference from finite samples in settings in which the conditional distribution of the response is arbitrary and completely unknown. This research has inspired a number of follow-up works, such as \cite{Tib,Barber,Su1,Su2, CHH16}. There are several other feature selection methods that offer some level of FDR control (e.g. \cite{BHq1,BHq2,Miller84,Miller02,GSell13,Liu10,Meinshausen10}). We refer to \cite{Candes} for a thorough comparison between the knockoff filter and these other approaches.

\subsection{Pseudo knockoff filter}
In this paper, we propose a pseudo-knockoff filter that inherits some advantages of the original knockoff filter and have greater flexibility in constructing their pseudo-knockoff matrix.
The first condition that we impose on the pseudo knockoff matrix is the following orthogonality condition:
   \begin{equation}\label{pko:ort}
          (X + \tilde{X})^T (X - \tilde{X} )  = 0   .
 \end{equation}
 It can be shown that this condition is equivalent to $X^TX = \td{X}^T \td{X}, \ X^T \td{X} = \td{X}^T X$.

We consider three classes of pseudo knockoffs that have different additional constraint. For the first class of pseudo knockoff filters, the pseudo knockoff matrix $\tilde{X}$ is chosen to be orthogonal to $X$, i.e. $X^T\tilde{X}=\tilde{X}^T X=0$. We call this pseudo knockoff the \textit{orthogonal pseudo knockoff}. It maximizes the difference between the pseudo knockoff matrix $\tilde{X}$ and its original design matrix $X$. The orthogonality condition makes $X_j$ and its knockoff orthogonal regardless of the correlation structure of $X$.

 The second class of pseudo knockoff filters is called the \textit{block diagonal pseudo knockoff}. We begin by constructing a block diagonal matrix $\BB$ that satisfies the property $\BB \succeq \Sigma^{-1}$. We can then solve for $\tilde{X}$ from the relationship $\BB = 4[(X - \tilde{X})^T(X - \tilde{X})]^{-1} $ where $\BB = 2 \diag( S^{-1}_{11},S^{-1}_{22},...,S^{-1}_{kk})$. The condition \eqref{pko:ort} and $ 4[(X - \tilde{X})^T(X - \tilde{X})]^{-1}  = 2 \diag( S^{-1}_{11},S^{-1}_{22},...,S^{-1}_{kk})$ imply that
% \beq\label{pko:bd0-0}
 \[
 X^T X = \td{X}^T \td{X} , \quad  X^TX - X^T\td{X}   = \diag( S_{11},S_{22},...,S_{kk}).
 \]
 %\eeq

We construct $\BB$ by adapting it to the structure of $X$. One of the guiding principles is to make it as small as possible so that we maximize the difference between $X$ and $\td{X}$.

 The third class of the pseudo knockoff filter is called the general pseudo knockoff by constructing $\BB$ whose principal submatrices are diagonal. The construction is similar to the case when $\BB$ is a block diagonal matrix.

 \subsection{A half Lasso statistic}

 We propose to use a half penalized method to construct the statistics of our pseudo knockoff filter. More specifically,
  the pseudo knockoff statistic is based on the solution of the following half penalized optimization problem
 \begin{equation}\label{introhalf}
  \min_{\hat{\beta},\tilde{\beta}}  \frac{1}{2} || y- X \hat{\beta} -\tilde{X}\tilde{\beta}||_2^2 + P( \hat{\beta}+  \tilde{\beta} ),
 \end{equation}
 where $P(x)$ is an even non-negative and non-decreasing function in each coordinate of $x$. An important consequence of the orthogonality condition (\ref{pko:ort}) is that we can reformulate the half penalized problem into two sub-problems equivalently
 %, i.e the minimization problem (\ref{introhalf}) is equivalent to the following minimization problem:
 \beq\label{eq:subprob}
    \min_{{\hat{\beta}+\tilde{\beta}}}  \left\{\frac{1}{2} || \frac{X+\tilde{X}}{2} (\hat{\beta}^{ls} + \tilde{\beta}^{ls}-  \hat{\beta} -\tilde{\beta}) ||_2^2  + P( \hat{\beta}+ \tilde{\beta} ) \right\} + \min_{\hat{\beta}-\tilde{\beta} } \left\{ \frac{1}{2} || \frac{X-\tilde{X}}{2} (\hat{\beta}^{ls} - \tilde{\beta}^{ls}- ( \hat{\beta} -\tilde{\beta})) ||_2^2
    \right\},
 \eeq
 where $\hat{\beta}^{ls}$ and $\tilde{\beta}^{ls}$ are the least squares coefficients by regressing $y$ on the augmented feature matrix $[X,\;\tilde{X}]$.
If we choose $P= \lambda || \cdot ||_{l^1}$, we obtain a \textit{half Lasso} method. We will mainly focus on the half Lasso statistic in this paper.
 Once we solve the half penalized problem, we can construct the pseudo knockoff statistic as follows
%  \beq\label{pko:sta-0}
\[
    W_j \triangleq (\hat{\beta}_j + \tilde{\beta}_j) \cdot \sign( \hat{\beta}_j - \tilde{\beta}_j ) \textrm{ or } W_j =  \max{ (  | \hat{\beta}_j| , |\tilde{\beta}_j| ) }\cdot \sign( | \hat{\beta}_j| - |\tilde{\beta}_j | ).
    \]
%     \eeq
We then apply a procedure similar to the knockoff filter \eqref{def-T0} to select features.

We have carried out a number of numerical experiments for different design matrices with various correlation structures to test the performance of the three classes of pseudo knockoff filters and compare their performance with that of the knockoff filter. For the examples that we consider in this paper, our numerical experiments indicate that all three classes of pseudo knockoff filters with the half Lasso statistic have FDR control. Moreover, the orthogonal and the general pseudo knockoff filter seem to offer more power than that of the knockoff filter with the Lasso Path or the half Lasso statistic, especially when the features are highly correlated.

\subsection{Uniform FDP bounds}
There has been some recent progress in obtaining uniform FDP bounds in \cite{KatSab17,KatRam18}.
Using \eqref{def-T0}, \eqref{def:FDP} and \eqref{eq:mart}, we can divide the control of FDR into three steps. First of all, we construct an estimate of $FDP$. We then choose a data-dependent threshold $T$ that achieves some adaptivity. The final step is to obtain an estimate for $E[ FDP(T) / \wh{FDP(T)} ] $ for this adaptive threshold, $T$. In \cite{KatRam18}, the authors showed that the above strategy of controlling FDR provides a general strategy for a variety of existing procedures that offer FDR control under some assumptions. In \cite{KatSab17}, the authors established a uniform bound across all possible threshold for the knockoff filter
\beq\label{eq:kfunif}
E \lt[
\sup_{t > 0}  \f{ FDP(t)} { \wh{ FDP(t)}} \rt]  \leq  E\lt[ \sup_{t > 0}  \f{     \#\{   j : W_j  \geq t  \  \& \  \b_j = 0\}        }{    1+  \#\{   j: W_j   \leq -t  \ \& \ \b_j = 0             \}   }
  \rt] \leq 1.93.
\eeq
In \cite{KatRam18}, the above uniform FDP bounds are established for several FDR procedures under some independence assumption similar to the
\textit{ i.i.d signs for the nulls} in the knockoff filter.

Inspired by the work of \cite{KatSab17,KatRam18}, we establish a uniform FDP bound under an assumption weaker than the independence assumption on the conditional distribution of the statistic $W$. Specifically, we prove the following theorem.
  \begin{thm}\label{thm:main}
 Let $\cF$ be a $\sigma$ field that satisfies:  (a) $| W_i|$ is $\cF$ measurable for all null $i$; (b) conditional on $\cF$,  $W_{S_0}$ can be divided into $m$ groups $C_1, C_2,..., C_m\ (C_i \subset S_0)$ such that the elements of $\sign(W_{C_i} )$ are mutually independent with $P( \sign(W_j )=1) = P( \sign(W_j) = -1  ) = 1/2$ for $j \in C_i$. For any $t > 0$, we have
  \beq\label{pk:ineqexp}
   E \left[ \frac{ \# \{  j \in S_0 : W_j \geq t \} } {      \# \{  j \in S_0 : W_j \leq - t \} + m }  \Big|  \cF   \right] \leq 1 \; .
   \eeq
 Moreover, if $W_{S_0}$ further satisfies $W_{S_0}  \overset{d}{=} - W_{S_0}$ conditional on $\cF$, we have
 \beq\label{pk:unif}
   E \left[ \sup_{t > 0} \frac{ \# \{  j \in S_0 : W_j \geq t \} } {      \# \{  j \in S_0 : W_j \leq - t \} + m }  \Big|  \cF   \right] \leq 3.9.
 \eeq
  \end{thm}

 Although Theorem 1.1 does not provide FDR control for the pseudo knockoff filter, it provides some partial understanding of the pseudo knockoff filter. For the block diagonal and the general pseudo knockoff filters, we verify that the pseudo knockoff statistic $W_j$ satisfies the assumption in Theorem \ref{thm:main} for some appropriate $\sigma$ field $\cF$. For the orthogonal pseudo knockoff filter, the pseudo knockoff statistic $W_j$ does not satisfy the assumption in Theorem \ref{thm:main}. To gain some understanding of the orthogonal pseudo knockoff filter, we obtain a relatively tight upper bound for the distribution function of $\frac{ \# \{  j \in S_0 :  \ W_j \geq t \} } {    \# \{  j \in S_0 : \ W_j \leq - t \}  } $ for fixed $t$ when $\Sigma^{-1}$ is diagonally dominated or when $\Sigma^{-1}$ has some special structure.

The rest of the paper is organized as follows. In Section \ref{sec:pk}, we introduce the three classes of pseudo-knockoff filters and discuss some essential properties of the pseudo knockoff filters. In Section \ref{num}, we present a number of numerical experiments to demonstrate the effectiveness of the proposed methods. In Section \ref{sec:unif}, we prove \eqref{pk:ineqexp} and outline the proof of \eqref{pk:unif} in Theorem \ref{thm:main}. In Section \ref{sec:opk}, we provide some partial analysis of the orthogonal pseudo knockoff filter.

 \section{A pseudo knockoff filter}\label{sec:pk}
In this section, we describe how to construct the three classes of pseudo knockoff filters and the half Lasso statistic. We will also discuss some of the essential properties of these pseudo knockoff filters and the half Lasso statistic.
%In this section, we construct pseudo knockoffs with constraints weaker than \eqref{revko1} and the pseudo knockoff statistics that satisfy the assumption in Theorem \ref{thm:main}.

 \subsection{The Basic Constraint and a Symmetry Property}
Given a design matrix $X \in R^{n \times p}$ with $n > 2p$, the basic constraint of the pseudo knockoff matrix is given by
% The basic constraint of the pseudo knockoff matrix is given by
 \beq\label{pko:con}
    \tilde{X}^T \tilde{X} = X^T X, \quad X^T\tilde{X} = \tilde{X}^TX.
 \eeq
 We can prove that \eqref{pko:con} and \eqref{pko:ort} are equivalent. It is obviously that \eqref{pko:con} implies \eqref{pko:ort}. If \eqref{pko:ort} holds, we have $X^T X - \td{X}^T\td{X} = X^T \td{X} - \td{X}^TX$. Note that the right hand side is a symmetric matrix, while the left hand side is a skew-symmetric matrix. It follows that $X^T X - \td{X}^T\td{X}$ is symmetric and skew-symmetric. Thus we must have $X^T X - \td{X}^T\td{X} = 0$, which further implies $X^T \td{X} - \td{X}^TX=0$. These two equations establish \eqref{pko:con}. The orthogonality condition \eqref{pko:ort} is the foundation of the pseudo knockoff filter and leads to the conditional independence between the amplitude of the null statistic $|W_{S_0}|$ and its sign $\sign(W_{S_0})$.

 \paragraph{Least squares coefficients}
 Consider the least squares coefficients $(\hat{\b}^{ls}, \td{\b}^{ls})$ of regressing $y$ on the augmented design matrix $[X  \ \tilde{X}]$. It is easy to obtain that $(\hat{\beta}^{ls}+ \tilde{\beta}^{ls}, \hat{\beta}^{ls}- \tilde{\beta}^{ls})$ are the least squares coefficients of regressing $y = X\b+ \e$ on $\lt[ \frac{ X+\tilde{X}}{2}  \   \frac{ X-\tilde{X}}{2} \rt]$. Using the orthogonality condition (\ref{pko:ort}),
 %Consider the least squares coefficients of regressing $y$ on the augmented design matrix $[X \tilde{X}]:   y  \sim X\hat{\beta}^{ls} + \tilde{X} \tilde{\beta}^{ls} = \frac{ X+\tilde{X}}{2}  (\hat{\beta}^{ls}+ \tilde{\beta}^{ls}) +  \frac{ X-\tilde{X}}{2}  (\hat{\beta}^{ls}- \tilde{\beta}^{ls}) . $ Using $y = X \beta  + \epsilon$ and the orthogonality condition (\ref{pko:ort}),
 we have a simple expression of the least squares coefficients,
 \beq\label{pko:ls}
  \begin{aligned}
   \left(
   \begin{array}{c}
   \hat{\beta}^{ls} +\tilde{\beta}^{ls}- \beta \\
   \hat{\beta}^{ls} - \tilde{\beta}^{ls} - \beta\\
   \end{array}
   \right)
   =&
   \left(
   \begin{array}{c}
   \left[ (\frac{X+\tilde{X}}{2})^T \frac{X+\tilde{X}}{2} \right]^{-1}   (\frac{X+\tilde{X}}{2})^T\epsilon \\
   \left[ (\frac{X- \tilde{X}}{2})^T \frac{X- \tilde{X}}{2} \right]^{-1}   (\frac{X-\tilde{X}}{2})^T\epsilon \\
   \end{array}
   \right)
   \triangleq
   \left(
   \begin{array}{c}
   \epsilon^{(1)} \\
   \epsilon^{(2)} \\
   \end{array}
   \right)
   .\\
     \end{aligned}
     \eeq
 %  \end{equation}
 The above relationship will be used repeatedly throughout the paper.
   Denote
   \beq\label{def:nota}
  \bal
  %\epsilon^{(1)} \triangleq   \left[ \left(\frac{X+\tilde{X}}{2} \right)^T \frac{X+\tilde{X}}{2} \right]^{-1} \left(\frac{X+\tilde{X}}{2} \right)^T\epsilon ,   &\ \epsilon^{(2)} \triangleq   \left[ \left(\frac{X - \tilde{X}}{2} \right)^T \frac{X -\tilde{X}}{2} \right]^{-1} \left(\frac{X - \tilde{X}}{2}\right)^T\epsilon     \\
   \eta \triangleq \hat{\beta}^{ls} + \tilde{\beta}^{ls} = \beta + \epsilon^{(1)}, & \  \xi \triangleq \hat{\beta}^{ls} - \tilde{\beta}^{ls} = \beta+ \epsilon^{(2) }\; .
  \eal
   \eeq
   From the orthogonality property \eqref{pko:ort}, we know that  $(\frac{X+\tilde{X}}{2})^T\epsilon$  and   $(\frac{X - \tilde{X}}{2})^T\epsilon$
   have independent multivariate normal distributions.  Using \eqref{pko:ls}, we know that $\epsilon^{(1)}$ and $\epsilon^{(2)}$, $\eta =\hat{\beta}^{ls} + \tilde{\beta}^{ls}$ and  $\xi = \hat{\beta}^{ls} - \tilde{\beta}^{ls} $ are also independent.

 \paragraph{The Pseudo Knockoff Statistics and Their Properties}
 %  \end{equation}
 %where $\hat{\beta}^{ls} $ and $\tilde{\beta}^{ls}$ are the least squares coefficients given in (\ref{pko:ls}).
According to \eqref{eq:subprob}, we can solve $\hat{\b} + \td{\b}$ and $\hat{\b}  - \td{\b}$ in the half penalized problem \eqref{introhalf} separately. Thus the solution can be expressed as
%  It is easy to derive that the solution of \eqref{introhalf} can be expressed as
 \beq\label{pko:sol}
 \hat{\beta} + \tilde{\beta} = f(\hat{\beta}^{ls} + \tilde{\beta}^{ls}) = f(\eta) , \quad \hat{\beta} - \tilde{\beta} = \hat{\beta}^{ls} - \tilde{\beta}^{ls} = \xi\; ,
 \eeq
 for some function $f : R^p \to R^p$. We construct the pseudo knockoff statistic as follows
   %\begin{equation}\label{psuha2}
   \beq\label{pko:sta}
    W_j \triangleq (\hat{\beta}_j + \tilde{\beta}_j) \cdot \sign( \hat{\beta}_j - \tilde{\beta}_j ) \textrm{ or } W_j \triangleq  \max{ (  | \hat{\beta}_j| , |\tilde{\beta}_j| ) }\cdot \sign( | \hat{\beta}_j| - |\tilde{\beta}_j | ).
 \eeq
 The pseudo knockoff statistic satisfies the following two properties.

 \textbf{Amplitude Property} \quad The amplitude of $W$ is determined by  $\hat{\beta}+\td{\beta} = f(\eta)$ and $|\hat{\beta} -\td{\beta} | =|\xi|$. In fact, using the definition of $W$ and \eqref{pko:sol}, we have
  \[
 |W| = |\hat{\beta} + \tilde{\beta} | = |f(\eta)| \quad \textrm{or} \quad |W| =  | \hat{\beta} | \vee | \tilde{\beta} |  =\frac{1}{2} (  | \hat{\beta}+\td{\beta}  + | \hat{\beta} -\td{\beta} | |  \vee  | \hat{\beta}+\td{\beta}  -  | \hat{\beta} -\td{\beta} |   |  ).
  \]

  \textbf{Sign Property} \quad The sign of $W$ is determined by $\sign(\hat{\beta} + \tilde{\beta})$ and $\sign(\hat{\beta} -\td{\beta})$. Since $\sign(| \hat{\beta}  | - | \tilde{\beta}| ) = \sign(| \hat{\beta}  |^2 - | \tilde{\beta}|^2 )$, for both definitions of $W$, we have
  \[
  \sign(W)  = \sign(\hat{\beta} + \td{\beta}) \cdot \sign(\hat{\beta} - \td{\beta} ) = \sign(f(\eta)) \cdot \sign(\xi).
  \]

  Now we show that the pseudo knockoff statistic satisfies a symmetry property.
 \begin{prop}\label{prop:sym}%(The Symmetry Property of the Pseudo Knockoff Statistic)
 Conditional on $\eta$, we have $W_{S_0}  \overset{d}{=} - W_{S_0}$,  where $S_0 \triangleq \{j: \beta_j = 0 \}$ and the pseudo knockoff statistic $W_j$ is defined in (\ref{pko:sta}). Consequently, for any threshold $t > 0$, we have
 \beq\label{pko:sym}
   \#\{ j :\beta_j = 0 \textrm{ and } W_j \geq t \} \overset{d}{ = }         \#\{ j :\beta_j = 0 \textrm{ and } W_j \leq -t \}  .\;
 \eeq
 \end{prop}

   \begin{proof}
 According to \eqref{def:nota} and \eqref{pko:sol}, the solution of the half penalized problem can be expressed as
% \beq\label{pko:sym11}
 \[\hat{\beta} + \tilde{\beta} =f(\eta) =  f(\beta+ \epsilon^{(1)}), \quad \hat{\beta} - \tilde{\beta} = \xi = \beta + \epsilon^{(2)}.\]
 %\eeq

 Next, we replace $(\epsilon^{(1)}, \epsilon^{(2)})$ by $( \epsilon^{(1)}, - \epsilon^{(2)})$ to generate a new pair of solutions $(\hat{\beta}^{new}, \tilde{\beta}^{new})$. From \eqref{def:nota}, changing $\epsilon^{(2)}$ to $-\epsilon^{(2)}$ does not change $\eta$. Thus, we obtain
% \beq\label{pko:sym12}
% \begin{aligned}
\[ \hat{\beta}^{new} + \tilde{\beta}^{new} = f(\eta)  = \hat{\beta} + \tilde{\beta},  \quad \hat{\beta}^{new} -\tilde{\beta}^{new} =  \beta  - \epsilon^{(2)}.\]
% \end{aligned}
% \eeq
The amplitude and sign properties of $W$ imply $| W^{new}_{S_0} | = |W_{S_0}| $ and
 \[
   \sign(W^{new}_{S_0}) = \sign(   ( f(\eta))_{S_0} \cdot  (- \e^{(2)})_{S_0} ) = - \sign( (f(\eta) )_{S_0} \cdot (\e^{(2)})_{S_0} ) = -\sign(W_{S_0}).
 \]
 %\[
 %  sign(W^{new}_{S_0}) = sign(   ( f(\eta))_{S_0} \cdot  (- \xi)_{S_0} ) = - sign( (f(\eta) )_{S_0} \cdot (\xi)_{S_0} ) = -sign(W_{S_0}).
 %\]
 Hence $W^{new}_{S_0} = -W_{S_0} .$

 %For the other choice of statistic given in (\ref{psuha2}), we have the same result $W_{S_0}^{new} =-W_{S_0} $.
  Recall that $W_{S_0}$ is generated by $\epsilon^{(1)},\epsilon^{(2)}$ and that $\epsilon^{(1)},\;\epsilon^{(2)}$ have independent multivariate normal distributions with zero mean. Conditional on $\eta$ (or equivalently $\epsilon^{(1)}$), we have
%   \beq\label{pko:sym2}
 \[  (\epsilon^{(1)},\epsilon^{(2)}) \overset{d}{=} (  \epsilon^{(1)},  - \epsilon^{(2)}) \Longrightarrow  W_{S_0} \overset{d}{=} W_{S_0}^{new} =  - W_{S_0} .\]
 %\eeq
\eqref{pko:sym} is a directly result of $W_{S_0}  \overset{d}{=} -W_{S_0}$.
 \end{proof}

\paragraph{A half Lasso statistic}
 We assume that $n > 2p$ and choose $P(x) = \lambda ||x||_1$ in \eqref{introhalf} to obtain a \textit{half Lasso} optimization problem:
 \vspace{-0.05in}
 \beq\label{eq:half}
\min_{\hat{\beta},\tilde{\beta}} \frac{1}{2} || y- X \hat{\beta} - \tilde{X}\tilde{\beta}||_2^2 + \lambda  ||  \hat{\beta}+ \tilde{\beta} ||_1 .
     \vspace{-0.05in}
 \eeq
 We then define the pseudo knockoff statistic according to \eqref{pko:sta}. It satisfies the symmetry property in proposition \ref{prop:sym}. We have conducted many simulations with different design matrices and signal sparsity and found that the half Lasso statistic offers robust performance when the tuning parameter $\lam$ is of the same order as the noise level. Thus we can choose the tuning parameter $\lambda$  by $\lambda = \mu  || U^T y||_2 / \sqrt{n-2p}$,  where $U \in R^{n \times (n-2p)}$ is an orthonormal matrix such that $[X \tilde{X}]^TU =0$. In fact, $U^Ty$ is exactly the residue of regressing $y$ onto $[X \ \tilde{X}]$. From our numerical study, we also observe that the power of the half Lasso statistic is not very sensitive to $\mu$ for a small range of $\mu$ centered at $\mu = 1$  and the numerical results seem to suggest that $\mu = 0.75$ is among the optimal choice. Thus we choose $\lambda = 0.75 || U^T y||_2 / \sqrt{n-2p}$ as the default tuning parameter. One can verify the symmetry property of the pseudo knockoff statistic using a similar argument.

  \subsection{Construction of the Pseudo Knockoff Matrix}\label{pko:conmat}
 In the previous subsection, we described the basic constraint \eqref{pko:con} for the pseudo knockoff matrix. In this subsection, we impose an additional constraint on $\tilde{X}$ so that we can obtain another important property for the pseudo knockoff statistic. In particular, we are interested in three classes of pseudo knockoff matrices, namely the \textit{orthogonal}, the \textit{block diagonal} and the general pseudo knockoff matrices.

 From \eqref{pko:ls} and \eqref{def:nota}, we know that the covariance matrix of $\epsilon^{(2)}$, or equivalently $\xi$, is given by
 \beq\label{pko:con0}
 \BB \triangleq  4[(X - \tilde{X})^T(X - \tilde{X})]^{-1} \; .
 \eeq
  We can design $\BB$ in such a way that we obtain some special correlation structure on $\xi$. To increase the power of the pseudo knockoff filter, we would like to construct $\td{X}$ such that the difference between $\td{X}_j$ and $X_j$ is large. Since  $||X_j - \td{X}_j||_2^2 = (( \BB/ 4 )^{-1})_{jj}$,
   %In order to do so,
   we aim to design $\BB$ as small as possible. Due to the existing constraint \eqref{pko:con} or \eqref{pko:ort}, the covariance matrix $\BB$ cannot be chosen arbitrarily. We give a necessary and sufficient condition on $\BB$ to find $\td{X}$ that satisfies \eqref{pko:ort} and \eqref{pko:con0}.

 \paragraph{Necessary Condition on $\BB$}
 Assume that there exists some $\td{X}$ that satisfies \eqref{pko:ort} and \eqref{pko:con0} and $X-\tilde{X}$ has full rank. Performing SVD on $(X - \td{X}) / 2 $, we have  $(X -  \tilde{X})/2 = \PP \MM^{-1} $  for some orthonormal matrix $\PP \in R^{n\times p}$ and some invertible matrix $\MM \in R^{p \times p}$. As a result, we get $\BB = [ ( \PP \MM^{-1})^T (\PP \MM^{-1})]^{-1}  = \MM \MM^T$ and $\tilde{X} = X - 2 \PP \MM^{-1}$.
 Substituting the last equation into the orthogonal condition $(X+\tilde{X})^T (X -\tilde{X}) = 0$ (see \eqref{pko:ort}), we obtain

 \[
 \bal
 4(  X -    \PP \MM^{-1} )^T \PP \MM^{-1} = 0   &\quad  \iff \quad  \MM^{-T} \MM^{-1} =  \MM^{-T} \PP^T X   \\
 \iff \quad  \MM^{-1}  = P^T X  \ & \quad  \Longrightarrow  \quad \BB = (X^T P P^T X)^{-1} .\\
 \eal
 \]
 Since $P \in R^{n\times p}$ is orthonormal, we have
 \beq\label{pko:con_B}
 X^T P P^T X \preceq X^T \II  X = X^T X = \Sigma \quad  \Longrightarrow   \quad \BB  =  (X^T P P^T X)^{-1} \succeq \Sigma^{-1}\;.
 \eeq
 \paragraph{Sufficiency} If $B$ satisfies \eqref{pko:con_B}, we have $B - \Sigma^{-1} \succeq 0$ and can construct $\td{X}$ as follows
 %We prove that the condition \eqref{pko:con_B} on $\BB$ is sufficient to find a $\td{X}$ that satisfies \eqref{pko:ort} (or \eqref{pko:con}) and \eqref{pko:con0}. To see this, we construct $\td{X}$ as follows
 \beq\label{pko:con4}
 \td{X} = X ( \II - 2 \Sigma^{-1} \BB^{-1})  +2  \UU \CC \BB^{-1} \;,
 \eeq
 where $\CC \in R^{p \times p}$ satisfies $\CC^T \CC = \BB - \Sigma^{-1} $  and  $\UU \in R^{n\times p}$ is an orthonormal matrix with $\UU^T X =0$. We will show that $\td{X}$ constructed from \eqref{pko:con4} satisfies \eqref{pko:ort} and \eqref{pko:con0} in the end of Appendix A.

 \subsubsection{An Orthogonal Construction}
The simplest construction is to choose $\BB = 2 \Sigma^{-1}$, which is equivalent to the following
%In this special pseudo knockoff filter, we impose an orthogonal constraint on the second condition of \eqref{pko:con}
 \beq\label{pko:con2}
 X^T X = \td{X}^T \td{X}, \quad  X^T \td{X} =\td{X}^T X =  0\;.
 \eeq
% to yield the \textit{orthogonal pseudo knockoff} matrix.
We call this special pseudo knockoff the \textit{orthogonal pseudo knockoff} since $\td{X}$ and $X$ are orthogonal. To construct an orthogonal pseudo knockoff matrix $\tilde{X}$, we first find the SVD of $X\in R^{n\times p}: \ X = U D V^T, \; U \in Orth^{n \times p}, \; D =\diag\{\sigma_1,...,\sigma_p\}$ and $V \in Orth^{p \times p} .$
 We then choose any orthonormal matrix $W\in R^{n\times p}$, whose column space is orthogonal to that of $X$ (i.e. $X^TW=0$), and construct the pseudo knockoff matrix $\tilde{X}$ as $ \tilde{X} = WDV^T$. It is easy to verify that $\tilde{X}$ satisfies \eqref{pko:con2}.
%the pseudo matrix condition (\ref{pko:con2}).

 \subsubsection{A Block Diagonal Construction}
 \paragraph{A Block Diagonal Construction}
 Consider a block diagonal matrix $\BB = 2 \diag( S^{-1}_{11},S^{-1}_{22},...,S^{-1}_{kk})$, where $S_{ii}$'s are invertible matrices. The constraint on $\BB$ is equivalent to
 \beq\label{pko:bd1}
 2 \BB^{-1} =  \diag( S_{11},S_{22},...,S_{kk})  \preceq 2 \Sigma .
 \eeq
 Hence  $(X - \tilde{X})^T (X - \tilde{X}) = 4 \BB^{-1} =  2 \diag( S_{11},S_{22},...,S_{kk})$. Using this relationship together with the basic constraint \eqref{pko:con}, i.e. $X^TX = \td{X}^T \td{X},  X^T \td{X} = \td{X}^T X$, we obtain
 \beq\label{pko:bd0}
 X^T X = \td{X}^T \td{X} , \quad  X^TX - X^T\td{X}   = \diag( S_{11},S_{22},...,S_{kk}).
 \eeq
 Assume that  $X$ can be clustered into $(X_{G_{1}}, X_{G_2},...,X_{G_k})$. Inspired by the group knockoff construction in \cite{Barber}, we first choose $S_{ii} \triangleq  \gamma  \Sigma_{G_i, G_i} = \gamma X_{G_i}^T X_{G_i} , \ i =1,2,...,k .$
 %S_{ii} \triangleq  \gamma  \Sigma_{G_i, G_i} = \gamma X_{G_i}^T X_{G_i} , \ i =1,2,...,k .
 The constraint \eqref{pko:bd1} implies $\gamma \cdot \diag(\Sigma_{G_1,G_1},\Sigma_{G_2,G_2},...,\Sigma_{G_k,G_k}) \preceq 2\Sigma$. In order to maximize the difference between $X$ and $\tilde{X}$, $\gamma $ should be chosen as large as possible:
 $  \gamma \leq  \min \{1,2\cdot \lambda_{\min} (D\Sigma D) \}$, where $D=\diag(\Sigma^{-1/2}_{G_1,G_1},\Sigma^{-1/2}_{G_2,G_2},...,\Sigma^{-1/2}_{G_k,G_k}   ) $. To ensure that the matrix $(X + \td{X})^T(X + \td{X})$ is nonsingular, we choose $\gamma = \frac{1}{1.2} \min \{1,2\cdot \lambda_{\min} (D\Sigma D) \}  $ in our numerical experiments. Once we construct $\BB$, we can generate the pseudo knockoff matrix via the procedure described earlier. This construction is useful if the features $X_j$ are clustered.

 \subsubsection{A general construction} \label{sec:gen}
 In general, we first divide the features $X_j$ into $m$ groups $C_1, C_2,.., C_m$ such that the correlation within each group is relatively weak. We remark that this criterion of partition is different from the grouping strategy in the block diagonal construction. The motivation of this partition is that $(\Sigma^{-1})_{C_jC_j}$ may be close to a diagonal matrix, which can be useful for the later construction of $\BB$.

We give two examples to illustrate why this partition may give rise to $(\Sigma^{-1})_{C_jC_j}$ that is close to a diagonal matrix.
 %can take advantage of the special structure of $\Sigma^{-1}$ to construct an effective matrix.
%comes from the following two examples.}
For example, if each $X_j$ is only strongly correlated with its neighbors $X_{j+i}$ for $|i|$ small, we can choose $C_k = \{ im + k:  i=0,1,..  \}$ for $k=1,2,..,m$. If $\Sigma_{ij} = X_i^TX_j = \rho^{|i-j|}$ for some $\rho >0$, $\Sigma^{-1}$ is tridiagonal and thus $(\Sigma^{-1})_{C_j, C_j}$ is a diagonal matrix. Another example is that if $X$ can be clustered into several groups such that the within-group correlation is stronger than the between-group correlation and the maximal group size is $m$, then we can pick $C_i$ as the $i$th element in each group for $1 \leq i \leq m$. If the between group correlation is $0$,  $(\Sigma^{-1})_{C_j, C_j}$ is also a diagonal matrix.

We construct a diagonal matrix $\bS_j$ that majorizes $(\Sigma^{-1})_{C_j, C_j}$ using a semidefinite program (SDP)
 \[
 \textrm{ minimize  } trace(\bS_j)  \quad \textrm{ subject to }  \quad \gamma (\Sigma^{-1})_{C_j, C_j}  \preceq \bS_j,  \quad 2 \leq (\bS_j)_{ii} \; .
 \]
The above SDP is similar to the SDP in the knockoff construction \cite{Candes} and can be solved very efficiently. $\gamma > 1$ is some parameter to be determined. If $(\Sigma^{-1})_{C_j, C_j}$ is close to a diagonal matrix, we can construct a $S_j$ such that their entries are not too large.
 %it enables us to construct a $\bS_j$ whose entries are not too large.
Next, we construct $\BB$ as follows
\beq\label{eq:gener}
\BB_{C_i , C_i} = \bS_i,  \quad  \BB_{C_i, C_j} = \gamma ( \Sigma^{-1})_{C_i, C_j}  \quad   1 \leq i \ne j  \leq m .
\eeq
 The difference between $\BB$ and $\Sigma^{-1}$ is on the diagonal. The above $\BB$ satisfies constraint \eqref{pko:con_B}
 \[
 \bal
 &\BB - \gamma \Sigma^{-1} = \diag(  B_{C_1, C_1} - \gamma  ( \Sigma^{-1} )_{C_1, C_1}  \ ,..., \ B_{C_m, C_m} - \gamma  ( \Sigma^{-1} )_{C_m, C_m}  )  \\
  =& \diag(  \bS_1 - \gamma  ( \Sigma^{-1} )_{C_1, C_1} \  , ..., \  \bS_m - \gamma  ( \Sigma^{-1} )_{C_m, C_m}  )   \succeq  0 \quad  \Rightarrow \quad  \BB \succeq \gamma \Sigma^{-1} \succeq \Sigma^{-1} .
 \eal
 \]
 We choose $\gamma = 1.2$ to ensure that $(X + \wt{X})^T (X+\wt{X})$ is nonsingular.

Among three constructions of the pseudo knockoff matrix, we choose the general construction as the default construction. After we construct the pseudo knockoff matrix $\td{X}$, we use $y, [X \ \td{X}]$ to calculate the half Lasso statistic and finally apply the knockoff+ filter \eqref{def-T0} with the target FDR level $q$ to selection features.

\paragraph{Relation to the knockoff filter}
If $m = 1$, $\td{X}$ constructed via the block diagonal or the general construction is exactly a knockoff matrix of $X$ \cite{Candes}. The constraint \eqref{revko1} in the original knockoff filter implies that $(X -\td{X})^T(X-\td{X})$ is a diagonal matrix, which in turn forces $[(X -\td{X})^T(X-\td{X})]^{-1} = \BB/4$ to be a diagonal matrix. In the construction of the pseudo knockoff matrix \eqref{pko:bd1} or \eqref{eq:gener},  we only require that $\BB$ be a block diagonal matrix or some submatrices of $\BB$ be diagonal. In this case, we can consider the pseudo knockoff filter as a generalization of the knockoff filter.

By comparing our block diagonal pseudo knockoff construction with the group knockoff filter in \cite{Barber}, we can see that the pseudo knockoff matrix, $\td{X}$, in \eqref{pko:bd0} is actually a group knockoff matrix of $X$. The group knockoff filter is originally designed for group selection with group FDR control while our block diagonal pseudo knockoff filter is designed for feature selection.

 \section{Numerical results for the pseudo knockoff filter}\label{num}
 In this section, we perform a number of numerical experiments to test the robustness of the pseudo knockoff filter and study the performance of various methods.

\subparagraph{Notations.}
$\beta_i \overset{i.i.d}{\sim} \{ \pm A\}$ means that $\beta_i$ takes value $A$ or $-A$ independently with equal probability $1/2$. We denote the orthogonal pseudo knockoff, the pseudo knockoff with the block diagonal construction, and the pseudo knockoff with general construction as \textit{orthogonal (OPK), block diagonal (BDPK), general (GPK)} pseudo knockoff.

 \subparagraph{Data} %Throughout all simulations, the noise $\epsilon \in R^n $ is a standard Gaussian, i.e. $\epsilon \sim N(0, I_n)$.
 Given some covariance matrix $\Sigma$, we first draw the rows of the design matrix $X \in R^{n \times p}$ from a multivariate normal distribution $N(0,\Sigma)$, and then normalize the columns of $X$. The pseudo knockoff matrix is generated according to Section \ref{pko:conmat}. To generate the signal strength $\beta\in R^{p}$, we choose $k$ coefficients $\beta_{i_1}, \beta_{i_2},..., \beta_{i_k}$ randomly and set $\beta_{i_j} \overset{i.i.d}{\sim} \{ \pm A\}$. Finally, the response variable $y\in R^n$ is generated from $y = X \beta + \epsilon, \e \sim N(0, I_n) $. Unless we specify otherwise, we will use the following default setup, i.e. the sample size is $p = 500, \; n =1500$,  the sparsity is $k = 30$, the signal amplitude is $A =3.5$ and the covariance matrix is $\Sigma = I_p$.

 \subparagraph{Methods}
 The methods that we focus on include the OPK, BDPK and GPK filters
% the orthogonal, the block diagonal and general pseudo knockoff filters
with the half Lasso statistic ($\lambda = 0.75$). We use the knockoff+ filter \eqref{def-T0} with nominal FDR level $q = 20\%$. We assume that every 5 features form a group and then construct the BDPK matrix. We choose $C_k=\{im + k: i =0,1,.. \}$ with $m=2,3, 5$ to construct the GPK matrix.
% For general pseudo knockoff matrix with parameter $m=2,3,5$, we choose $C_k=\{im + k: i =0,1,.. \}$.
After obtaining the fitted value $\hat{\b},\td{\b}$ in the half Lasso problem, we have two choices to construct the statistic, $W$, in \eqref{pko:sta}. Denote $W^{(1)}_j = (\hat{\beta}_j +  \tilde{\beta}_j )  \cdot \sign( \hat{\beta}_j   -  \tilde{\beta}_j  ) $ and  $W^{(2)}_j =  |\hat{\beta}_j | \vee|\tilde{\beta}_j  | \cdot \sign( | \hat{\beta} _j| - |\tilde{\beta}_j |) $. For the OPK,  we use $W^{(2)}$, which seems to offer more power with OPK;  for other pseudo knockoff filters, we consider both constructions of $W$ in \eqref{pko:sta}. There are 9 methods in total.

 \subsection{Numerical evidence of FDR control for the pseudo knockoff filter}\label{fdr-ctr}
 In this subsection, we perform extensive numerical experiments to test whether the pseudo knockoff filter has FDR control. For this purpose, we apply it to select features in the linear model $y = X\beta +\epsilon$ with different design matrices under various extreme conditions.
 %We test the orthogonal pseudo knockoff, banded and block diagonal pseudo knockoff and pay particular attention to the FDR control, the power and the expectation (mean of the ratio defined in \eqref{pko:exp}).

 The default simulated data is discussed at the beginning of Section \ref{num} and we vary one of the default settings in each experiment as follows (one setting is varied while keeping the others unchanged).

 \indent \textit{(a) Sparsity}: $k$ varies from $10,20,30,...,90, 100$. \\
 \indent \textit{(b) Signal amplitude}: $A$ varies from $2.8,2.9,...,4.2$. \\
 \indent \textit{(c) Correlation Structure}: We use the covariance matrix $\Sigma \in R^{500 \times 500}, \Sigma_{ij} = \rho^{|i-j|}$ and vary the correlation level $\rho =0,0.1,...,0.9$.  \\
 \indent \textit{(d) The sample size}: We vary the sample size $n = 150l,\; p = 50 l$ and sparsity $k=10l$ with $ l \in\{2,3,...,12\}$.   \\
 \indent \textit{Group Structure}: %We consider a design matrix $X \in R^{1500 \times 500}$ with a group structure. Specifically,
 We assume that the features $X_j$ can be clustered into $100$ groups with $5$ features in each group. To generate a different group structure, we choose the covariance matrix  $\Sigma_{ii} =1$, $\Sigma_{ij} = \rho$ for $i \ne j$ in the same group and $\Sigma_{ij}= \gamma \cdot\rho$ for $i\ne j$ in different groups and generate the design matrix $X$ as in the previous discussion.  \\
 \indent \textit{(e) The within-group correlation}: $\gamma =0$ is fixed and $\rho$ varies from $0,0.1,0.2,...,0.9$. \\
 \indent \textit{(f) The between-group correlation}:  $\rho = 0.5$ is fixed and $\gamma$ varies from $0,0.1,0.2,...,0.9$.

% We test the orthogonal, banded, and block diagonal pseudo knockoff filters and pay particular attention to
 We pay particular attention to  the \textit{FDR} (the mean false discovery proportion), the \textit{power} (the expected proportion of true discoveries) and the \textit{expectation}, which is defined as the expectation of
%  \beq\label{pko:exp}
 $\frac{ \# \{ j:  \ W_j \geq T  \ \&  \ \beta_j = 0   \}   } {  \# \{  j:  \ W_j \leq - T \ \& \ \beta_j = 0\}   +1     }.$
Each experiment is repeated 200 times to calculate these quantities.
%the mean FDR, the mean power, and the expectation.
The design matrix $X$ and the pseudo knockoff matrices $\td{X}$ are fixed over these trials. We plot the results of OPK and BDPK (m=5),
GPK (m=2) with $W^{(1)} = (\hat{\beta} +  \tilde{\beta} )  \cdot \sign( \hat{\beta}   -  \tilde{\beta}  ) $ in Figure \ref{fdr1}, \ref{fdr2}.

 The dotted line in Figure \ref{fdr1} and Figure \ref{fdr2} represents the prescribed FDR $q$ or constant $1$ as a reference. In all figures, we observe that the FDR is controlled by $q=20\%$. From the results of the expectation, we observe that all of them are close to or less than $1$.
Other six methods described before Section \ref{fdr-ctr} control FDR in the above examples. In Section \ref{sec:FDR}, we will provide partial analysis to gain some understanding of the pseudo knockoff filter.
%we will provide partial analysis why the pseudo knockoff filter could provide FDR control.

  \begin{figure}[H]
   \centering
      \includegraphics[width =\textwidth ]{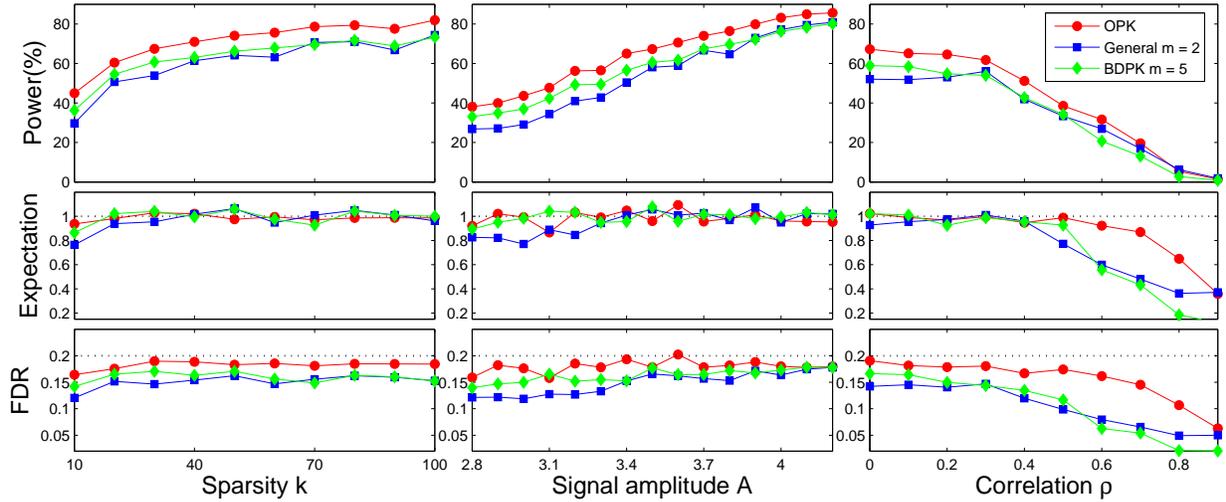}
   \caption{Testing the orthogonal, the block diagonal and the general pseudo knockoff+ at a nominal FDR $q= 20\%$ by varying the
 sparsity, the signal amplitude, or the feature correlation. }
%   \caption{Testing the orthogonal, the banded, and the block diagonal pseudo knockoff+ at a nominal FDR $q= 20\%$ by varying
% sparsity, the signal amplitude, or the feature correlation. }
   \label{fdr1}
   \end{figure}
   \begin{figure}[h]
   \centering
        \includegraphics[width =\textwidth ]{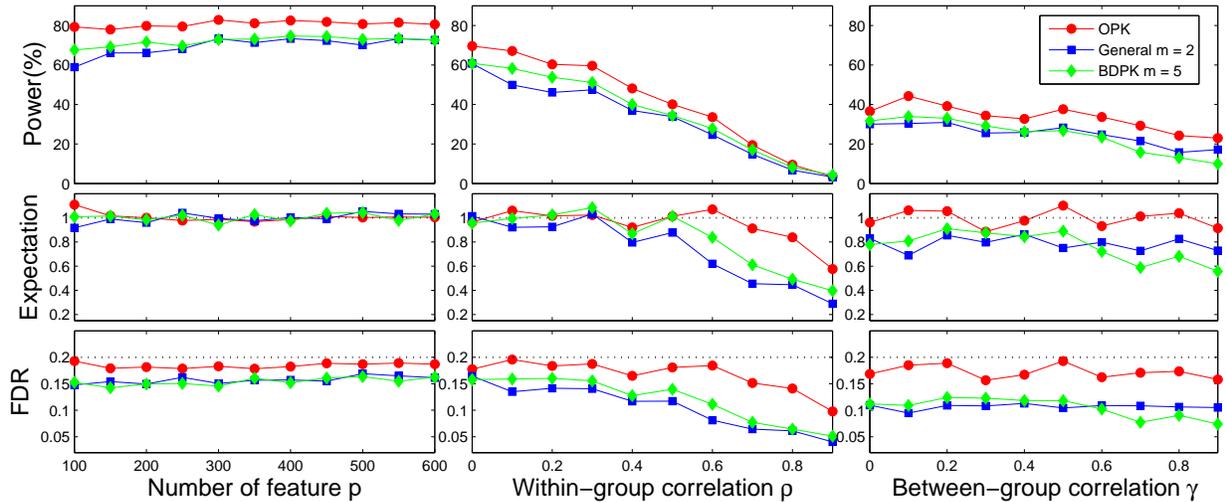}
  % \caption{Testing the orthogonal, the banded, and the block diagonal pseudo knockoff+  at a nominal FDR $q= 20\%$ by varying one of the following: the number of features $p$, the within-group correlation, and the between-group correlation. }
\caption{ Testing several pseudo knockoff+ filters  at a nominal FDR $q= 20\%$ by varying the number of features $p$, the within-group correlation or the between-group correlation. }
   \label{fdr2}
   \end{figure}

 \subsection{The pseudo knockoff filter in some correlated scenarios}\label{sec:comp}
 Due to the constraints on the knockoff matrix in the original knockoff filter, strongly correlated features force the $s_i$ to be small \cite{CHH16}, which may lead to loss of some power. A main advantage of the pseudo knockoff filter is that it relaxes the constraint of $\td{X}$ in \eqref{pko:con}. In some correlated scenarios with some special structure, we can construct the pseudo knockoff matrix that is adapted to such structure and improve the power. To illustrate the effectiveness of the pseudo knockoff filter, we compare the knockoff filter using various statistics with various pseudo knockoff constructions using the half Lasso statistic.

\paragraph{Statistics}
We use the half Lasso statistic with $\lambda = 0.75  || U^T y||_2 / \sqrt{n-2p}$ ($n > 2p$) for the pseudo knockoff filter. We also consider the corresponding statistics in the knockoff filter for comparison. Specifically, we consider the knockoff filter with the half Lasso or Lasso using the same tuning parameter ($\lambda = 0.75  || U^T y||_2 / \sqrt{n-2p}, \; n > 2p$) and the sign max statistic $W^{(2)}$.
In addition, we have tested the knockoff filter with other statistics, including the Lasso path and the OMP statistics.
The knockoff matrix is generated by the SDP construction introduced in \cite{Candes}.  In the following examples,  we use a slightly larger signal amplitude $A =5$. For these methods, we use the knockoff+ filter \eqref{def-T0} with nominal FDR level $q = 20\%$. Throughout all the examples in this Section, we repeat the experiment 200 times to obtain the FDR and the averaged power.

\paragraph{Group Structure}
We consider a design matrix $X \in R^{1500 \times 500}$ with a group structure and  two sparsity cases: $k=30$ and $k=100$. In particular, we consider experiment (e) in Section \ref{fdr-ctr}. The within-group correlation factor $\rho$ varies from $0.5,0.55,,...,0.95$ and the between-group correlation factor is $\gamma = 0$. In all other settings, we use the default values. By taking advantage of the {\it a priori} knowledge of the correlation structure of $X$, we construct %{\color{red} BDPK and GPK pseudo knockoff matrices}
the BDPK and GPK
with $m=5$. We also implement the OPK with $W^{(2)}$ statistic for comparison.
%and plot the result of general pseudo knockoff with $W^{(1)}, W^{(2)}$ statistics below.}
\begin{figure}[H]
 \centering
 \includegraphics[width =\textwidth ]{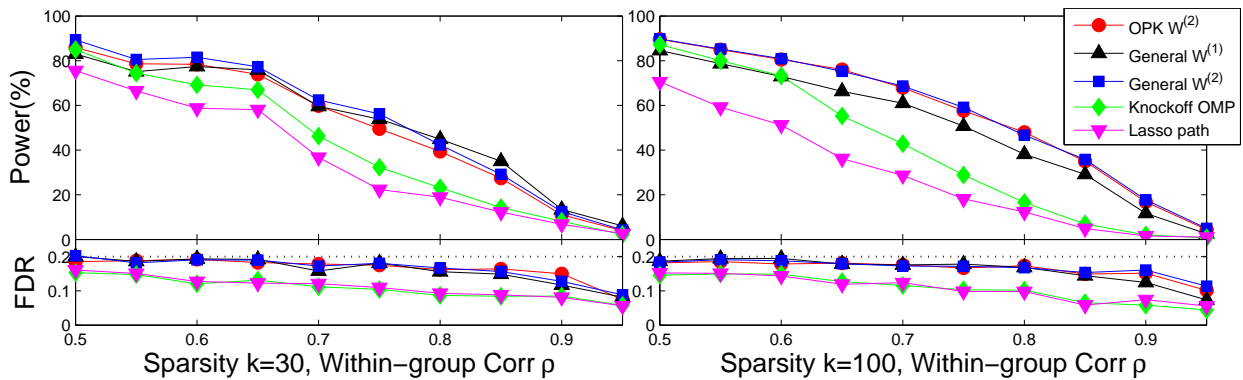}
 % \caption{Comparing the orthogonal, the block diagonal pseudo knockoff filter and the knockoff filter with the OMP statistic at nominal FDR $q= 20\%$ by  varying the within-group correlation.}
 \caption{Comparing the orthogonal, the general pseudo knockoff+ filter and the knockoff+ filter with several statistics at nominal FDR $q= 20\%$ by  varying the within-group correlation. Here, the general  $W^{(i)}$  means the method using the general pseudo knockoff construction and $W^{(i)}$ statistic.}
  \label{blockdiag}
 \end{figure}

%In both figures, the pseudo knockoff filters control FDR and outperform the knockoff filter with the Lasso path statistic. The BDPK with $W^{(1)}$ statistic (not plotted here) also outperforms the knockoff filter with the Lasso path statistic but offers less power than that of the OPK or the GPK.
In both figures, the pseudo knockoff filters control FDR and outperform the knockoff filter with the OMP or the Lasso path statistic. The BDPK with $W^{(1)}$ statistic (not plotted) also outperforms the knockoff filter with two statistics but offers less power than that of the OPK or the GPK. %The knockoff filter with the Lasso sign max statistic $W^{(2)}$ or with the half Lasso statistic ($W^{(1)}$ version) (not plotted) offers more power than the OMP or the Lasso path statistics. Their power are comparable to that of the OPK or the GPK.

 \paragraph{Decaying Structure}
 We consider a design matrix $X \in R^{1500 \times 500}$ with some decaying structure and  two sparsity cases: $k=30$ and $k=100$. Specifically, the design matrix $X$ is generated from $N(0,\Sigma)$ with $\Sigma_{ij} = \rho^{|i-j|}$, where $\rho$ varies from $0.5,0.55,...,0.95$. Other settings use the default values. We know \textit{a priori} that the off-diagonal elements of $\Sigma^{-1}$ decay rapidly. Thus, we apply the GPK with parameter $m=5$. We also implement the OPK with $W^{(2)}$ statistic for comparison.
\begin{figure}[h]
 \centering
   \includegraphics[width =\textwidth ]{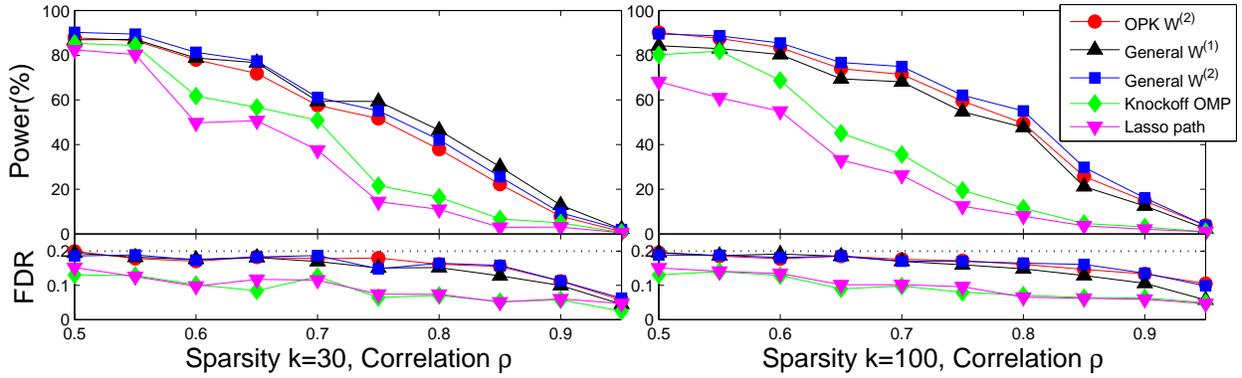}
    \caption{Comparing the orthogonal, the general pseudo knockoff+ filter and the knockoff+ filter at nominal FDR $q= 20\%$ by varying the pairwise correlation.}
 \label{blockdiag}
 \end{figure}

In Figure \ref{blockdiag}, we again observe that in both figures the pseudo knockoff filters control FDR  and outperform the knockoff filter with the OMP or the Lasso path statistic. We also implement the GPK with parameter $m=2$ and two statistics $W^{(1)}$ and $W^{(2)}$. Its performance is still better than that of the knockoff filter with the OMP or the Lasso path statistic.

In these two examples with group or decaying structure, the knockoff filter with the Lasso sign max statistic $W^{(2)}$ or with the half Lasso statistic ($W^{(1)}$ version) offers more power than that of the OMP or the Lasso path statistic. Their powers are comparable to that of the OPK or the GPK. The tuning parameter $\lambda = 0.75  || U^T y||_2 / \sqrt{n-2p}$, which was designed for the pseudo knockoff filter with half Lasso statistic, works equally well for the knockoff filter with the Lasso or the half Lasso statistic in these two examples.

\paragraph{Exploring the special structure in the precision matrix} %consider a dual correlation structure of the last two examples and
Next, we investigate how we can design an effective pseudo knockoff filter by taking advantage of the special structure in the precision matrix $\Sigma^{-1}$. We consider three examples :  (a) $(\Sigma^{-1})$ is a block diagonal matrix with equal  block size $5$ and $(\Sigma^{-1})_{ii} = 1$, $ (\Sigma^{-1})_{ij} = \rho$ for $i \ne j$ in the same block and $0$ otherwise; (b)  $(\Sigma^{-1})_{ij} =  \rho^{  |i - j|}$;  (c) $(\Sigma^{-1})_{ii} =1$ and $(\Sigma^{-1})_{ij} = \rho$ for $ i \ne j$. We then generate $X$ from the multivariate normal distribution $N(0, \Sigma)$ as in the previous numerical examples.
%Note that the correlation structure in example (a), (b) can be regarded as a dual version of the group structure and the decaying structure.
We vary $\rho$ from $0.5,0.55,...,0.95$ in example (a), (b) and from $0,0.1,0.2,.., 0.9$ in example (c). We consider the sparsity level $k=30$ and focus on the pseudo knockoff filter with the half Lasso statistic and the knockoff filter with the Lasso and the half Lasso statistics. The special structure of the precision matrix suggests that choosing $m=5$ for the GPK would be a reasonable choice for these examples. We also implement the OPK for comparison.

\begin{figure}[H]
 \centering
   \includegraphics[width =\textwidth ]{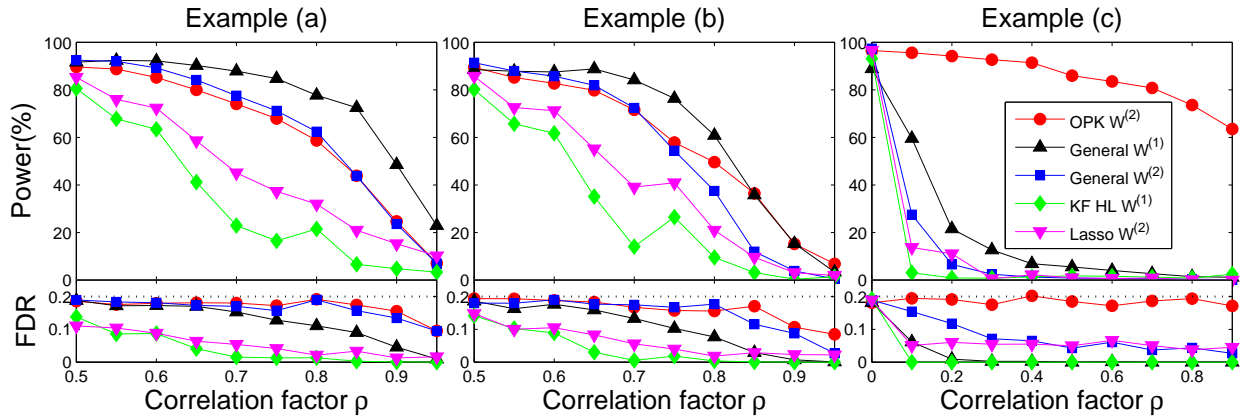}
    \caption{Comparing the pseudo knockoff+ filter with the knockoff+ filter with several statistics at nominal FDR $q= 20\%$ by varying $\rho$ in various precision matrices. The left, middle and right subfigures correspond to example (a), (b) and (c), respectively. \textit{KF HL} is short for the knockoff filter with the half Lasso and \textit{Lasso} is short for the knockoff filter with the Lasso statistic.}
 \label{precision}
 \end{figure}
We observe that when we construct the knockoff matrix $\td{X}$ using the original knockoff filter, the difference between some feature $X_i$ and its knockoff $\td{X}_i$ generated by the SDP construction is very small for some cases in example (b), (c) when $\rho$ is large. We compute the mean $s_i$ (see \eqref{revko1}) in example (c) for 10 different values of $\rho$ that we use in this example. Their mean values for $\rho=0,0.1,...,.0.9$ are  $0.426, 0.031,  0.013,  0.007,    0.005,    0.003 ,   0.0019 ,   0.0013 ,   0.0007,    0.0003$, respectively. In our computation, we have used the glmnet package in Matlab \cite{glmnet} to solve the Lasso optimization problem, $(\hat{\b}, \td{\b}) = \arg \min_{( \hat{b}, \td{b}) } \f{1}{2}  || y-  X \hat{b} -\td{X} \td{b}  ||_2^2 + \lambda  || (\hat{b}, \td{b}) ||_1$. The original results that we have obtained are a bit surprising in the sense that the Lasso statistic constructed this way fails to control FDR in this extreme example. To gain some understanding what goes wrong, we found that the numerical solution of this Lasso optimization problem is significantly different from the numerical solution of $(\hat{\b}, \td{\b}) = \arg \min_{ ( \hat{b}, \td{b}) } \f{1}{2}  || y-  \td{X}\td{b}  - X\hat{b}  ||_2^2 + \lambda  || (\hat{b}, \td{b}) ||_1$, which is the same Lasso optimization problem except that we have swapped the order of the input variables $(X, \td{X})$. This numerical error may be attributed to the extremely small difference between $X_i$ and $\td{X}_i$ for some $i$ and the degeneracy of the augmented design matrix $[X \  \td{X} ]$. %This degeneracy of the augmented design matrix
This numerical error may lead to the violation of the flip-coin property of the knockoff statistic $W$ constructed from the numerical solution $(\hat{\b}, \td{\b})$, which may explain why we could lose FDR control in this extreme case. To overcome this difficulty, we turn off the knockoff $\td{X}_i$ for $X_i$ if $s_i$ is small when we construct the knockoff Lasso sign-max statistic. More specifically, we define an index set, $P \teq \{ i : s_i \geq 0.001\} $. We first solve $(\hat{\b}, \td{\b}_P) = \arg \min_{\hat{b}, \td{b} } \f{1}{2}  || y-  X\hat{b}  - \td{X}_P \td{b}_P||_2^2 + \lambda  || (\hat{b}, \td{b}_P ) ||_1$. We then construct  $W^{(2)}_P =( | \hat{\b}_P| \vee  |\td{\b}_P| ) \cdot  \sign(   | \hat{\b}_P|  -| \td{\b}_P|) $ and set $W^{(2)}_{P^c} = 0$. The numerical results that we present in Figure 5 for the Lasso $W^{(2)}$ statistic are obtained using this slightly modified procedure in constructing the knockoff Lasso statistic.

In three subfigures in Figure \ref{precision}, the OPK and the GPK with the half Lasso statistic control FDR and outperform the knockoff filter with the half Lasso statistic $W^{(1)}$ (the half Lasso with $W^{(2)}$ offers less power than the half Lasso with $W^{(1)}$) and the Lasso sign max statistic. The Lasso with $W^{(1)}$ statistic offers performance similar to that of $W^{(2)}$. We have implemented the knockoff filter with the OMP and the Lasso path statistics in example (c) and found that these statistics perform poorly, which may be attributed to the smallness of $s_i$ in this example. In general,  from $ 0 \prec   \diag(s) \preceq 2 X^T X \  \Rightarrow \  (X^T X)^{-1} \preceq 2 (\diag(s))^{-1} $ , one can show that the slow decay of the off-diagonal elements of $ (X^T X)^{-1} $ forces $s_i$ to be extreme small, which could lead to a significant loss of power of the knockoff filter. The OPK with the half Lasso statistic maintains a high power in example (c), which may be attributed to the orthogonal property between $X$ and its pseudo knockoff $\td{X}$. We have also tested the OPK with the least squares statistic in example (c). Due to the slow decay of the off-diagonal elements of $(X^T X)^{-1}$, $\sign(W_j^{ls}) \  1  \leq j \leq p$ are correlated for large $\rho$ and we found that the least squares statistic fails to control the FDR in these cases.

In these examples, we find that in the sparse case, the GPK with $W^{(1)}$ offers more power than the GPK with $W^{(2)}$, while in the non-sparse case, $W^{(2)}$ offers more power than $W^{(1)}$.
 %In these two examples, we find that in the non-sparse case, the GPK with $W^{(2)}$ offers more power than the GPK with $W^{(1)}$.
 In Section \ref{sec:unif}, we show that the GPK with $W^{(1)}$ statistic satisfies the assumptions in Theorem \ref{thm:main}. Although we cannot verify these assumptions for $W^{(2)}$ statistic due to the fact that $|W_{S_0}| $ and $\sign(W_{S_0})$ are not independent, we expect that Theorem \ref{thm:main} is approximately true for $W^{(2)}$ due to the sign property $\sign(W^{(1)}) = \sign(W^{(2)})$ and the similarity between $W^{(1)}$ and $W^{(2)}$.

\section{Some analysis of the pseudo knockoff filter}\label{sec:FDR}

In this section, we will provide some partial analysis for the pseudo knockoff filter, which may provide some understanding regarding the performance of the pseudo knockoff filter.
\subsection{A uniform FDP bound}\label{sec:unif}
 In the knockoff filter, the following expectation inequality
 \begin{equation}\label{psumar}
% \beq\label{ko:ineq}
 E\left[   \frac{   \# \{ j :W_j \geq T,\beta_j =0  \} } { \#\{ j : W_j \leq -T ,\beta_j =0   \}  +1  }  \right] \leq 1,
 \end{equation}
 plays an important role in obtaining the exact FDR control of the knockoff filter.

The numerical experiments in Section \ref{num} show that the pseudo knockoff with the half Lasso statistic offers FDR control and the expectation \eqref{psumar} is approximately valid. Since we relax one of the constraints in the knockoff filter, we cannot apply the supermartingale argument to obtain \eqref{psumar} for the pseudo knockoff filter. To gain some understanding why \eqref{psumar} may be valid for the pseudo knockoff with the half Lasso statistic, we would like to estimate the expectation \eqref{psumar} for fixed $t$ and the suprema over all $t$ in Theorem \ref{thm:main}.
 For a technical reason, we still cannot prove \eqref{psumar} right now. Instead we prove a weaker version of \eqref{psumar} by replacing 1 in the denominator by $m$.

According to the assumption of $W_{S_0}$ in Theorem \ref{thm:main}, in the extreme but highly unlikely case, $W_{S_0}$ can be $m$ copies of $(\eta_1, \eta_2,..,\eta_L)$ where $\eta_j$ are independent and symmetric random variables. Then \eqref{pk:unif} reduces to \eqref{eq:kfunif} with a upper bound that is about twice as large as the upper bound in \eqref{eq:kfunif} and \eqref{pk:ineqexp} reduces to  $E \left[ \frac{ \# \{  j : \  \eta_j \geq t \} } {      \# \{ j: \   \eta_j \leq - t \} + 1 }  \Big|  \cF   \right] \leq 1$. Since $\sign( \eta_j)$ are i.i.d Rademacher random variables, the latter expectation is $1 - 2^{-n}$, where $n  = \# \{  j : |\eta_j| \geq t\}$.  Both results in Theorem 1.1 are relatively tight. For the half Lasso statistic, this extreme scenario is very unlikely to occur since the $l^1$ regularization imposes sparsity and forces  $\hat{\b}_j+\td{\b}_j$ to be zero for many features $X_j$ in a correlated group. As a result,  $W_j$ is zero for many features $X_j$ in a correlated group and thus it is very unlikely that such an extreme scenario can be realized for the half Lasso statistic. In  Section \ref{sec:comp}, we consider some highly correlated examples, including the cases with $0.95$ within-group correlation and with $0.95$ correlation between $X_i$ and $X_{i+1}$ for each $i$. These highly correlated examples in principle could generate strongly correlated $W_{S_0}$, but we observe that the pseudo knockoff filter with the half Lasso statistic still offers FDR control.

  \begin{proof}[Proof of \eqref{pk:ineqexp}]
  Let $  N_t  \triangleq  \{ j \in S_0: | W_j| \geq t\} $. By assumption of $\cF$, $N_t$ is determined and we can divide $N_t$ into $m$ groups $C_1, C_2,..., C_m\ (C_i \subset S_0)$ such that the elements of $\sign(W_{C_i} ) $ are mutually independent. Obviously, $|N_t| = \sum_{i=1}^m |C_i| $.  Using the following Cauchy-Schwarz inequality
 \[
  \sum_{i=1}^{m } \frac{a_i^2}{b_i } \sum_{i=1}^m b_i \geq   (\sum_{i=1}^m a_i)^2 \iff \frac{1}{\sum_{i=1}^m a_i}  \sum_{i=1}^{m } \frac{a_i^2}{b_i } \geq \frac{ \sum_{i=1}^m a_i   } {\sum_{i=1}^m b_i}, \ a_i ,b_i >0 \; ,
 \]
 with $a_i = |C_i| + 1, b_i = \# \{ j \in C_i : W_j \leq -t   \}  +1 $, we obtain
   \begin{align}
  & E  \left[ \frac{ \# \{  j \in S_0 : W_j \geq t \} } {      \# \{  j \in S_0 : W_j \leq - t \} + m }   \Big\vert \cF  \right] +1   = E  \left[  \frac{ |N_t|+ m} { \sum_{i=1}^{m} (     \# \{   j \in C_i : W_j \leq - t \} + 1 ) }   \Big\vert \cF  \right]  \notag  \\[2mm]
    \leq &  E\left[  \frac{1}{| N_t|+m} \sum_{i=1}^{m}    \frac{  ( |C_i| +1  )^2         }{  \# \{ j \in C_i : W_j \leq -t   \}  +1      }   \Big\vert \cF  \right]    =  \sum_{i=1}^{m}    \frac{|C_i| + 1 }{| N_t|+m}E\left[    \frac{   |C_i| +1          }{  \# \{ j \in C_i : W_j \leq -t   \}  +1      }   \Big\vert \cF  \right]   \notag \\[2mm]
  =& \sum_{i=1}^{m}   \frac{|C_i| + 1 }{| N_t|+m}       \left\{   1  + E \left[    \frac{     \# \{ j \in C_i : W_j \geq t   \}     }     {  \# \{ j \in C_i : W_j \leq -t   \}  +1   }    \Big\vert \cF \right] \right\} .\label{pko:ineq14}
  \end{align}
In the above derivation, we have used $\# \{ j \in C_i : W_j \leq -t   \}  +1  +  \# \{ j \in C_i : W_j \geq t   \}   = |C_i| + 1$ to obtain the first and the last equalities, and used the fact that $| N_t|$ and $ |C_i|$ are measurable with respect to $\cF$ to yield the second equality. From the assumption (b), $\one_{ W_j >0}$ with $j\in C_i$ are mutually independent and each obeys a binomial distribution. We yield
  \vspace{0.05in}
  \[
  E \left[    \frac{     \# \{ j \in C_i : W_j \geq t   \}     }     {  \# \{ j \in C_i : W_j \leq -t   \}  +1   } \Big\vert  \cF   \right]
   =  E \left[    \frac{     \# \{ j \in C_i : W_j  > 0  \}     }     {  \# \{ j \in C_i : W_j < 0  \}  +1   } \Big\vert  \cF   \right]
 = 1 - 2^{-|C_i|} \leq 1.
   \vspace{0.05in}
 \]
  Therefore, the last line in \eqref{pko:ineq14} is bounded by
 \[
  \frac{1}{ |N_t|+m} \sum\nolimits_{i=1}^{m}  2 ( |C_i| +1 ) = \frac{2}{|N_t|+m} \cdot (|N_t|+m) = 2 .
 \]
 Subtracting $1$ on both sides of \eqref{pko:ineq14} concludes the proof of \eqref{pk:ineqexp}.
  \end{proof}

The proof of \eqref{pk:unif} is more technical and we need the following concentration inequality.
\begin{lem}\label{lem:con}
Assume that the $\sigma$ field $\cF$ satisfies the conditions in Theorem \ref{thm:main} and $|W_{S_0}|$ are in decreasing order : $| W_{i_1} |\geq | W_{i_2} |\geq .. \geq |W_{i_{l }}| > 0 $, where $W_{i_k }, \ 1 \leq k \leq l$ are all nonzero elements in $W_{S_0}$. Denote $V_j^{\pm} = \# \{ i_k : (\pm) W_{i_k}  \geq |W_{i_j}|     \} =  \# \{ i_k : (\pm) W_{i_k}  >0 \ \& \  k \leq j \}$. For any $ t > 1$ and $ i < j  \leq ti $, we have
\vspace{0.05in}
\beq\label{eq:con}
\bal
P\lt(  \f{   V_j^+      }{  V_i^-  + m    }  >  t   \rt)  & \leq \inf_{ \th > 0 }   \exp\lt( -\th  \lt( \f{ t \cdot  i - j}{2}  + tm\rt)  \rt) \cdot \lt( \f{   \exp(  m \th / 2 ) + \exp(  - m\th / 2)  }{2}   \rt)^{ (  j - i ) / m}       \\
&\cdot \lt( \f{   \exp( (1+t) m \th / 2 ) + \exp(  -(1+t) m\th / 2)  }{2}   \rt)^{    i   / m}   .
 \eal
  \vspace{0.05in}
 \eeq

\end{lem}
Roughly speaking, the above probability decays exponentially fast with respect to $i$ and $t$. To prove \eqref{eq:con}, we first apply the H\"{o}lder inequality to decouple correlated terms and then establish a bound of the moment generating function (MGF) of $V_j^+ + t V_i^+$ similar to the Heoffding MGF bound. Finally we apply the Laplace transform method. We will use \eqref{eq:con} and a slicing method to control the suprema in \eqref{pk:unif}. We defer the proof of \eqref{pk:unif} and Lemma \ref{lem:con} to Appendix A.

Next, we show that the pseudo knockoff statistic satisfies the assumptions in Theorem \ref{thm:main}.
%As a result, we can obtain the FDR control of the pseudo knockoff filter.

 \paragraph{Independence of $\xi$}
 Let $m$ be the largest block size of $\BB$ in the block diagonal construction or the parameter in the general construction. Recall that the covariance matrix of $\xi = \hat{\beta}^{ls} - \td{\beta}^{ls}$ is $\BB$. Since $\BB_{C_j, C_j}$ is a diagonal matrix in the general construction, thus $\xi_i, i \in C_j$ are mutually independent.

 For the block diagonal construction, we can choose $C_j$ to  be the collection of the $j$-th element in each block if there exists such an element. Then $\xi_i, i \in C_j$ are also mutually independent.

\paragraph{The general construction}
For $\td{X}$ generated by the general construction, we choose $W_i =   (\hat{\b}_i + \td{\b}_i) \cdot \sign(  \hat{\b}_i - \td{\b}_i ) = f(\eta)_i \sign(\xi_i) $ \eqref{pko:sta}. Let $\cF$ be the $\sigma$ field generated by $\eta $. According to the amplitude property of $W$, $|W|$ is $\cF$ measurable. Since $\eta$ and $\xi = \b+ \e^{(2)}$ are independent and $\xi_i , i \in C_j \cap S_0$ are mutually independent, we conclude that
\[  \sign( W_i) =    \sign( f(\eta)_i) \cdot \sign ( \b_i + \e^{(2)}_i  ) =  \sign( f(\eta)_i) \cdot \sign(\e^{(2)}_i),  \quad  i \in C_j \cap S_0,
\]
are symmetric and mutually independent conditional on $\cF$. This  verifies condition (b) in Theorem \ref{thm:main}. The additional condition $W_{S_0} \overset{d}{=} -W_{S_0}$ follows from Proposition \ref{prop:sym}.

 In the numerical experiments that we presented in Section \ref{num}, we have also used $W = |\hat{\beta}| \vee | \td{\beta} | \cdot \sign(  |\hat{\beta}|  -  | \td{\beta} | )$. Although we cannot prove that this statistic satisfies the assumption in Theorem \ref{thm:main}, our numerical experiments seem to suggest that the FDR control is not sensitive to the choice of statistic in \eqref{pko:sta}.

 \paragraph{Block diagonal construction}
 If $\td{X}$ is generated by the block diagonal construction, we show that both statistics in \eqref{pko:sta} satisfy the assumptions in Theorem \ref{thm:main}. For $W_i =   (\hat{\b}_i + \td{\b}_i) \cdot \sign(  \hat{\b}_i - \td{\b}_i )$, we can use the same argument as above. % to show that $W$ satisfies the assumption in Theorem \ref{thm:main}.
  For $W_i = |\hat{\beta}_i| \vee | \td{\beta}_i | \cdot \sign(  |\hat{\beta}_i|  -  | \td{\beta}_i | )$, $\cF$ is the $\sigma$ field generated by $\eta$ and $|\xi|_{S_0}$. The amplitude property implies $|W_i|$ is $\cF$ measurable for null $i$.  The symmetry property of $W_{S_0} $ follows from Proposition \ref{prop:sym}.  It remains to verify that conditional on $\cF$, $\sign(W_i)$ are mutually independent for $i \in C_j \cap S_0$.

Note that $\Var(\xi) = \BB = \diag(S_{11}, S_{22}, ...,S_{kk}   )$, $\xi_{S_0} =\epsilon^{(2)}_{S_0}$ and the elements of $C_i$ come from different blocks.
 We can change the sign of $\epsilon^{(2)}_{S_0}$ in any block $S_{i_1i_1}, S_{i_2i_2},...,S_{i_ji_j}$ without changing $|\xi_{S_0}|$ and the joint distribution of $\epsilon^{(2)}_{S_0}$. Consequently, conditional on $\cF$, $\sign(\xi_i)$ are mutually independent for $i \in C_j \cap S_0$.
Using the independence of $\sign(\xi_{C_j \cap S_0})$, the sign property and the symmetry property of $W_{S_0}$, we verify the condition (b) in Theorem \ref{thm:main}.

 \subsection{Partial analysis of the orthogonal pseudo knockoff}\label{sec:opk}
 From the previous numerical results, we observe that the orthogonal pseudo knockoff is among the most powerful pseudo knockoffs and still maintains robust FDR control. One of the main reasons is that $\td{X}_j$ in OPK is orthogonal to $X_j$ and thus the difference between them is maximized.
In this subsection, we provide some partial analysis of the orthogonal pseudo knockoff with $W^{(1)}$ statistic and expect that similar results also hold for OPK with $W^{(2)}$ statistic. First we discuss several properties of the orthogonal pseudo knockoff.
% From the previous numerical results, we observe that the orthogonal pseudo knockoff offers more power than other pseudo knockoff filters and still maintains robust FDR control. In this section, we analyze OPK with $W^{(1)}$ statistic and expect that simialr result of holds for OPK with $W^{(2)}$ statistic.
%In this section, we perform further numerical experiments to gain some insight and provide some partial analysis.

% \subsection{Property of the Orthogonal Pseudo Knockoff}
 \paragraph{Symmetry Property} Since $X^T \td{X} =0$ is symmetric, the symmetry property stated in Proposition \ref{prop:sym} holds for the orthogonal pseudo knockoff.

% \paragraph{Control of the ratio}
 Recall $W_j  = (\hat{\b}_j + \td{\b}_j) \sign(\hat{\b}_j - \td{\b}_j) = f(\eta) \sign(\xi)$. We introduce the following notations %, which will be used frequently later.
 \beq\label{def:nota2}
 \bal
 & \Sigma = X^T X ,  \quad D = \diag(\Sigma^{-1}) = \diag(d_1,d_2,..,d_p) , \quad \wt{\Sigma^{-1}} = D^{-1/2} \Sigma^{-1} D^{-1/2}. \\
%&{\color{blue}  Z_1(t) \triangleq \# \{  j \in S_0 : W_j  \geq t \}, \quad Z_2(t) \triangleq \# \{  j \in S_0 : W_j  \leq -t \} ,  \quad \textrm{ for } t > 0 } \\
 \eal
 \eeq
 By definition, we have $(\wt{\Sigma^{-1}})_{ii}= 1$.   Let $\cF$ be the $\sigma$ field generated by $\eta$. Conditional on $\cF$, $|W_{S_0}|$ is determined. We assume that $|W_{S_0}|$ is arranged in a decreasing order and use the same notation $V_i^{\pm}$ as in Lemma \ref{lem:con}. Similar to \eqref{pk:ineqexp} or \eqref{eq:con}, we estimate the ratio $V_i^+/ V_i^-$.
% Similar to \eqref{pk:ineqexp} and \eqref{eq:con}, we estimate the ratio $Z_1(t) / Z_2(t)$. }
 \begin{thm}\label{thm:ort}
% For any $\delta \in (0,1)$ and $t > 0$, the orthogonal pseudo knockoff  satisfies
 For any $\delta \in (0,1)$ and $ j \geq 1$, conditional on $\cF =\sigma(\eta)$, the OPK satisfies
 \begin{align}
 &  P\left(\frac{V^+_j}{  V_j^- } \geq \frac{1+\delta }{1 -\delta} \  \Big| \cF \right)  \leq \frac{( 1 + 3\pi )  \lam_{\max} (\wt{\Sigma^{-1}}_{N_j N_j})   }{ \pi \delta^2 j } \; ,  \label{pko:ort1}% |N_{\eta}|^{-1} ,
 \end{align}
 where $N_j \teq \{  i_k : |W_{i_k}|  \geq |W_{i_j} | \}$ and $ \lam_{\max} (\wt{\Sigma^{-1}}_{N_j N_j})$ is the largest eigenvalue of the submatrix $ \wt{\Sigma^{-1}}_{N_j N_j}$.
 \end{thm}
  \begin{remark}
  Note that the diagonal elements of $\wt{\Sigma^{-1} }_{N_j N_j}$ are all $1$ and $|N_j| = j$. Note that $j = \tr(\wt{\Sigma^{-1}}_{N_j N_j}) = \sum_{i=1}^j \lambda_i(\wt{\Sigma^{-1}}_{N_j N_j} )$ and $\lambda_i (\wt{\Sigma^{-1}}_{N_j N_j}) >0$. Thus,  we  have $ \lambda_{\max} (\wt{\Sigma^{-1}}_{N_j N_j}) <j$.
 \end{remark}

From the sign property of $W$, we know $\sign(W_{S_0}) = \sign( (f(\eta))_{S_0}) \cdot \sign(\xi_{S_0})$. Denote $Y_i = \one_{W_i > 0}$. We first analyze the covariance of each pair $(Y_i, Y_j) , i,j\in S_0$.
 \begin{lem}\label{lem:cov}  Conditional on $\eta$, for any null variable $i,j$, we have %$Cov(Y_i, Y_j | \eta) \leq \frac{1}{2} | ( \wt{\Sigma^{-1}})_{ij} |$, where  $\wt{\Sigma}_{ij}$ is defined in \eqref{def:nota2}.
 \beq\label{lem:cov2}
 \Cov(Y_i, Y_j | \eta) \leq \frac{1}{2\pi} (\wt{\Sigma^{-1}})_{ij}(\one_{ ( f(\eta))_i > 0 } - \one_{ ( f(\eta))_i < 0 })    (\one_{ ( f(\eta))_j > 0 } - \one_{ ( f(\eta))_j < 0 }) + \frac{3}{2} (\wt{\Sigma^{-1}})^2_{ij}\; .
 \eeq
% where  $(\wt{\Sigma}^{-1})_{ij}$ is defined in \eqref{def:nota2}.
 \end{lem}
 We will defer the proof to Appendix B.

 \begin{proof}[Proof of Theorem \ref{thm:ort}]
  According to the symmetry property (Proposition \ref{prop:sym}) of OPK, for $i \in N_j  (\subset S_0)$, we have
 \beq\label{pko:mean}
E(Y_i | \eta)  = E( \one_{W_i > 0} | \eta)  = 1/2 , \quad  E(V_j^+ | \eta)  = E( V_j^- | \eta) = j / 2,  \quad  V_j^+  + V_j^- =j .
 \eeq
Denote $w_i\triangleq  \one_{f_i(\eta) >0} - \one_{f_i(\eta) <0}$.  Using \eqref{lem:cov2} and $\wt{\Sigma^{-1}}_{N_j N_j} \preceq \lam_{\max} (\wt{\Sigma^{-1}}_{N_j N_j})  \II$, we obtain
% \beq\label{pko:vskf2-1}
\beq\label{pko:var1}
 \bal
\Var(V_j^+ | \eta) \leq& \sum_{s, t\in N_j} \Cov(Y_s, Y_t |\eta) \leq \sum_{s,t\in N_j} \frac{  1  } {2\pi}  ( \wt{\Sigma^{-1}} )_{st}w_s w_t + \frac{3}{2} ( \wt{\Sigma^{-1}} )_{st}^2  \\
 =& \frac{w^T_{N_j}  (\wt{\Sigma^{-1}})_{N_j N_{j}}   w_{N_{j}} }{2\pi}  + \frac{3}{2} \tr(   (\wt{\Sigma^{-1}})^2_{N_{j} N_{ j}}  )
 \leq  \f{  \lam_{\max} (\wt{\Sigma^{-1}}_{N_j N_j}) j }{2\pi}  +  \frac{3 \lam_{\max} (\wt{\Sigma^{-1}}_{N_j N_j}) j  }{2}  \; .
 \eal
 \eeq

 Conditional on $\eta$,  we apply \eqref{pko:mean}, \eqref{pko:var1}, and the Chebyshev inequality to yield
 \begin{align}
  P(V^-_j \leq (1-\delta) j/2 \big| \eta)  &= P(V_j^- \geq (1+\delta) j/2 \big| \eta) = \frac{1}{2} P(   |V_j^- - j/2| \geq \delta j /2 \big| \eta )  \notag \\
 & \leq \frac{ 2\Var( V_j^- \big| \eta) } {  (\delta j )^2} \leq  \frac{( 1 + 3\pi ) \lam_{\max} (\wt{\Sigma^{-1}}_{N_j N_j})    }{ \pi \delta^2j} \; . \label{pko:ort12}
 \end{align}
 The first identity holds since the symmetry property \eqref{pko:sym} implies that $V_j^- \overset{d}{=} V_j^+= j - V_j^-$. The estimate \eqref{pko:ort1} follows from integrating the last inequality in \eqref{pko:ort12}.
 \end{proof}

For some design matrices that have certain special structure in $\wt{\Sigma^{-1}}$, e.g. the design matrices to be considered in the next subsection, we can show using \eqref{pko:var1} that $\Var(V_j^+ | \eta) = O(j)$.
%and the upper bound in \eqref{pko:ort1} is $O(1/j)$.
Conditional on $\eta$, if $\sign(W_i) \ i \in N_j$ are independent,  which is true if we use the knockoff statistic, we have $\Var(V_j^+|\eta) = j / 4$. In this case, $\Var(V_j^+ |\eta)$ is the same order as that in the knockoff for all $j$.

 \subsection{Some Special Design Matrices}\label{sec:spemat}
 For some special design matrices, we can improve the estimate of $\Var(V_j^+ | \eta)$ in \eqref{pko:var1} and get better control of $V_j^+/ V_j^-$.  In our simulations, we observe that the OPK offers robust FDR control. We would like to offer a partial explanation of this phenomenon.
 %In our simulations, we observe that the OPK offers robust FDR control and works better than other pseudo knockoff filters and the knockoff filter with the OMP statistic in these scenarios. We would like to offer a partial explanation of  this phenomenon.

  \paragraph{A diagonally dominated case}
   Let $X \in R^{n \times p}$ and $\Sigma = X^T X$. We consider several classes of design matrices described below.

   (a) For any $i\ne j$, $\langle X_i , X_j \rangle \triangleq X_i^T X_j  = \rho, \  \rho \in [0,1 ) $.

   (b) Assume that $X$ can be clustered into $k$ groups, $X = (X_{C_1}, X_{C_2},...,X_{C_k})$.  The within-group correlation of group $i$ is $\rho_i$ for some $\rho_i \in [0, 1 )$ and the between-group correlation is zero.

   (c) The sizes of different groups are equal. The within-group correlation is $\rho$ and the between-group correlation is $\gamma \cdot \rho$.\\
 Case (a) corresponds to setting  (a), (b), (d) with $\rho=0$ in Section \ref{fdr-ctr}; case (b) and (c) correspond to setting (e) and (f) in Section \ref{fdr-ctr}. Denote $\EE \triangleq \Sigma^{-1}$ for convenience. From \eqref{def:nota2}, $(\wt{\Sigma^{-1}})_{ij}  = \EE_{ij}/ ( \EE^{1/2}_{ii} \EE^{1/2}_{jj} )$. For the design matrices described above, we can show that $\Sigma^{-1}$ is diagonally dominated, i.e. $\sum_{j \ne i} |  (\Sigma^{-1})_{ij} | < \Sigma^{-1}_{ii}$. The proof is a bit technical and tedious. We will omit the proof here. From Lemma \ref{lem:cov}, we have
   \beq\label{pko:var2}
   \Cov(Y_i, Y_j| \eta) \leq \frac{1}{2\pi} (\wt{\Sigma^{-1}})_{ij} w_i w_j  + \frac{3}{2}(\wt{\Sigma^{-1}})^2_{ij} \leq c_0 |(\wt{\Sigma^{-1}})_{ij} |, \quad  c_0 = \frac{1}{2\pi}+ \frac{3}{2}  < 2\; .
   \eeq
   Since $\Sigma^{-1}$ is diagonally dominated, we can improve the estimate of $\Var(V_j^+| \eta)$ in \eqref{pko:var1}
%\vspace{-0.4in}
   \[\bal
   &\Var(V_j^+\big| \eta)  \leq   \sum_{s, t \in N_{j}} c_0   | (\wt{\Sigma^{-1}} )_{st}|
   =c_0 \sum_{s, t \in N_{j}} \frac{ |  \EE_{st}  |     }{  \EE^{1/2}_{ss} \EE^{1/2}_{tt}          }   \leq  c_0 \left(   \sum_{s, t\in N_{j}}   \frac{ |  \EE_{st}  |     }{  \EE_{ss}  }  \right)^{1/2 }  \left(   \sum_{s,t\in N_{j}}   \frac{ |  \EE_{st}  |     }{  \EE_{tt}  }  \right)^{1/2 } \\[1mm]
 =&  c_0 \left(   \sum_{s \in N_{j}}  \frac{ 1} {  \EE_{ss} }\sum_{t\in N_{j}}  |  \EE_{st}  | \right)^{1/2 }
    \left(   \sum_{ t\in N_{j}} \frac{1}{\EE_{tt}} \sum_{ s \in N_{j}}   |  \EE_{st}  |    \right)^{1/2 }
    \leq c_0 \left( \sum_{s \in N_{j}} \frac{2 \EE_{ss}   } {\EE_{ss} }   \right)^{1/2}   \left(   \sum_{ t\in N_{j}} \frac{ 2   \EE_{tt} }{\EE_{tt}}        \right)^{1/2 }    =2c_0j.\\
  \eal
  \]
  Here, we have used $\EE_{st} = \EE_{ts}$ and the diagonal dominated assumption to yield $ \sum_{t \in N_{j} } |  \EE_{st} | \leq \sum_{ t=1}^p  |  \EE_{st} | \leq 2 \EE_{ss} .$ With this refined estimate of $\Var(V_j^+ | \eta)$, the upper bound in Theorem \ref{thm:ort} can be reduced to
$  \frac{ 2+6\pi} {\pi\delta^2 j }$.

  \paragraph{Exponentially Decaying Class}
  %Assume that $ (\Sigma^{-1})_{ii} \geq 1, \quad  | (\Sigma^{-1})_{ij}  |   \leq C \rho^{|i-j|} $ for $\rho \in [0, 1)$ and some constant $C$.  The design matrix in setting (c) in Section \ref{fdr-ctr} has a similar structure.
 Assume that $| (\Sigma^{-1})_{ij}  |   \leq C \rho^{|i-j|} $ for $\rho \in [0, 1)$ and some constant $C$. The design matrix in setting (c) in Section \ref{fdr-ctr} has a similar structure. One can prove that $(\Sigma^{-1})_{ii} \geq 1$ using the fact that $\Sigma_{ii}=1$ and $\Sigma$ is positive definite. By our assumption, we have $ | ( \wt{\Sigma^{-1}})_{ij} | \leq  | (\Sigma^{-1})_{ij}  | \leq C \rho^{|i-j|}$.
 Hence, we have $\lambda_{\max}(\wt{\Sigma^{-1}}) \leq || \wt{\Sigma^{-1}} ||_{l_1} \leq 2 C/(1-\rho)$. Denote $c_0 = (1+3\pi)/ (2\pi)$.
 Using \eqref{pko:var1} and Theorem \ref{thm:ort}, we yield
 \[
 \Var(V_j^+ \big| \eta) \leq c_0 \lambda_{\max}(\wt{\Sigma^{-1}}) j\leq  \frac{2c_0 C j  } {1- \rho} , \quad P\left( \frac{V_j^+}{V_j^-}  \geq \frac{1+\delta }{1- \delta} \Big| \cF \right) \leq \frac{ 4c_0C } {\delta^2(1- \rho) j } \; .
 \]
Therefore, for all the design matrices that we considered in Section \ref{fdr-ctr} (up to randomness), we have $\Var(V_j^+ | \eta) = O(j)$ for all $j$. This may offer some partial explanation why we observe robust FDR control of OPK in these examples.

 \section{Concluding remarks}
 In this paper, we proposed a pseudo knockoff filter for feature selection with correlated features. Both the block diagonal and the general pseudo knockoff (GPK) filters preserve some essential features of the original knockoff filter but offer more flexibility in constructing the knockoff matrix. We also proposed the orthogonal pseudo knockoff (OPK) filter. Our numerical experiments seem to suggest that the pseudo knockoff filters have FDR control in the numerical examples that we considered in this paper. Moreover, the OPK and GPK filters seem to offer more power than the knockoff filter with the Lasso Path and the half Lasso statistics in these examples, especially when the features are highly correlated. For the block diagonal and the general pseudo knockoff (GPK) filters, we provided an estimate for the expectation of the ratio $ \frac{ \# \{  j \in S_0 :  \ W_j \geq t \} } {      \# \{  j \in S_0 : \ W_j \leq - t \} + m } $ for any fixed threshold in \eqref{pk:ineqexp} and its suprema over all possible thresholds in \eqref{pk:unif} under weaker assumptions on the conditional distribution of the statistic. For the orthogonal pseudo knockoff filter, we provided some estimate of the distribution function \eqref{pko:ort1}. This estimate provides a relatively tight upper bound when $\Sigma^{-1}$ is diagonally dominated or when $\Sigma^{-1}$ has some special structure. Although our analysis does not lead to FDR control, it may offer some partial understanding of the pseudo knockoff filter.

We would like to emphasize that our understanding of the pseudo knockoff filter is still quite limited. In some extreme cases, we found that the orthogonal pseudo knockoff filter with the least square statistic fails to control FDR. Although we have better understanding of the OPK with the half Lasso statistic and obtained better theoretical results for the GPK, these results do not provide a satisfactory explanation for the robust performance of the pseudo knockoff filter with the half Lasso statistic that we observed numerically. In our future study, we would like to further investigate whether one can find some appropriate conditions on the design matrices under which we can obtain exact FDR control for the pseudo knockoff filter with the half Lasso statistic. This question seems to be extremely difficult. Some new method of analysis needs to be developed to give an affirmative answer to this question.

\section*{Appendix}
 \appendix
 \section{Proof of Theorem \ref{thm:main}}
 The derivations in this Appendix are conditional on $\cF$ and we drop the notation of conditional expectation for simplicity.
 \begin{proof}[ Proof of Lemma \ref{lem:con} ]
 We first estimate the moment generating function (MGF) of $V_j^+ +t V_i^+$ and then apply the Laplace transform method to establish \eqref{eq:con}.
 Denote $\xi_k = \one_{W_{i_k} > 0  } - 1/2$, $\lam_k = 1 +t$ for $ k \leq i$ and $\lam_k = 1$ for $ i < k \leq j$. Since $|W_{i_l} | $ is decreasing, we obtain
 \[
 V_j^+ + t V_i^+ -  \f{ j + ti }{2} = \# \{   k  \leq j :  W_{i_k} > 0 \}  + t \# \{   k  \leq i :  W_{i_k} > 0 \}  -  \f{ j + ti }{2}  = \sum_{ k  \leq j } \lam_k  \xi_k ,  \quad V_i^- =  i  - V_i^+.
 \vspace{-0.05in}
 \]
According to the assumption of $W_{S_0}$ in Theorem \ref{thm:main}, conditional on $\cF$, we can divide $W_{S_0}$ into $m$ groups $C_1,C_2,..,C_m$ such that $\sign(W_k), k \in C_l$ are independent. We can use the H\"older inequality to decouple correlated terms and estimate the MGF of $V_j^+  + t V_i^+ -  \f{ j + ti }{2} $ for any $\th > 0$ as follows
\vspace{-0.05in}
 \[
 \bal
G(\th)  &=   E \exp(  \th (  V_j^+  + t V_i^+   - ( j + ti ) / 2  ))
= E \exp\lt( \sum_{ l=1}^m  \sum_{  k \in C_l , \  k \leq j } \lam_k \xi_k \th  \rt)  \\[1mm]
&\leq  \prod_{l=1}^m \lt\{  E \exp\lt( \sum_{  k \in C_l  , \  k \leq j } m \lam_k \xi_k \th \rt)  \rt\}^{1/m}   = \prod_{l=1}^m  \prod_{ k \in C_l ,  \  k \leq j  } (  E \exp(  m \lam_k \xi_k \th) )^{1/m} ,
\eal
\]
where we have used the fact that $\xi_k = \one_{W_{i_k} > 0  } - 1/2, \ k \in C_l$ are independent to yield the last equality. The symmetry assumption of $\sign(W_j)$ in Theorem \ref{thm:main} implies $\xi_k \sim \{\pm 1/2 \}$. Using the definition of $\lam_k$, we obtain
\vspace{-0.05in}
\begin{align}
G(\th) &\leq \prod_{l=1}^m  \prod_{ k \in C_l ,  \  k \leq j  } (  E \exp(  m \lam_k \xi_k \th) )^{1/m} = \prod_{k=1}^j (  E \exp(  m \lam_k \xi_k \th) )^{1/m}  \notag \\ %[1mm]
   =&  \lt( \f{   \exp( (1+t) m \th / 2 ) + \exp(  -(1+t) m\th / 2)  }{2}   \rt)^{    i   / m}  \lt( \f{  \exp(  m \th / 2 ) + \exp(  - m\th / 2)  }{2}   \rt)^{ (  j - i ) / m}  . \label{eq:mgf}
\end{align}
%For $\th > 0 $, using Holder inequality gives
%  \begin{align}
%  &  E \exp(  \th (  V_j^+  + t V_i^+   - ( j + ti ) / 2  ))
% = E \exp\lt( \sum_{ l=1}^m  \sum_{  k \in C_l , \  k \leq j } \lam_k \xi_k \th  \rt) \leq  \prod_{l=1}^m \lt\{  E \exp\lt( \sum_{  k \in C_l  , \  k \leq j } m \lam_k \xi_k \th \rt)  \rt\}^{1/m} \notag\\
%   = &\prod_{l=1}^m  \prod_{ k \in C_l ,  \  k \leq j  } (  E \exp(  m \lam_k \xi_k \th) )^{1/m}  =  \prod_{l=1}^m  \prod_{ k \in C_l ,  \  k \leq j  }  \lt( \f{  \exp(  m\lam_k\th  / 2 ) + \exp(  - m \lam_k\th / 2)  }{2}  \rt)^{1/m} \notag \\
%    =&  \lt( \f{  \exp(  m \th / 2 ) + \exp(  - m\th / 2)  }{2}   \rt)^{ (  j - i ) / m}    \cdot \lt( \f{   \exp( (1+t) m \th / 2 ) + \exp(  -(1+t) m\th / 2)  }{2}   \rt)^{    i   / m}, \label{eq:mgf}
% \end{align}
%where we use that $\xi_k, \ k \in C_l$ are mutually independent and each is a half of a Rademacher random variable to obtain the second and the third identity.
%Note that

To simplify the notations, we define $B(x,y,t, \xi, s)$ as follows
\vspace{0.05in}
\beq\label{eq:bxy}
%B(x,y,t, \xi , s ) \teq
\exp\lt( -\xi  \lt(  \f{  tx- y }{2} + s   \rt) \rt)    \lt( \f{   \exp( (1+t) \xi / 2 ) + \exp(  -(1+t) \xi / 2)  }{2}   \rt)^{   x}  \lt( \f{   \exp(  \xi / 2 ) + \exp(  - \xi / 2)  }{2}   \rt)^{y - x}  \; .
%e^{ -\xi s}  \lt( \f{   \exp( (1+t) \xi / 2 ) + \exp(  -(1+t) \xi / 2)  }{2}   \rt)^{   x}  \lt( \f{   \exp(  \xi / 2 ) + \exp(  - \xi / 2)  }{2}   \rt)^{y - x}   .
\eeq
For fixed $t, \xi >0$, it is easy to verify that $B(x,y,t,\xi, s)$ is a monotonically decreasing function of $x ,s$ and an increasing function of $y$. Choosing $\xi = m\th, x =i/m$ and $y= j/m $, we can simplify \eqref{eq:mgf} as follows
\beq\label{eq:simp}
E \exp(  \th (  V_j^+  + t V_i^+   - ( j + ti ) / 2  )) = G(\th)
\leq \exp\lt(m \th \lt(   \f{t i - j}{2m} +     s\rt)    \rt)  B\lt( \f{i}{m}, \f{j}{m}, t,m\th, s\rt)
\eeq
for any  $s$. %Choosing $s = (ti - j)/(2m) + t$ in \eqref{eq:simp} ($ms = (ti - j)/2 + t$)
Using \eqref{eq:simp} and applying the Markov inequality to $\exp(  \th (  V_j^+  + t V_i^+   - ( j + ti) /2   ))$ for any $\th>0$, we yield
\vspace{-0.1in}
\beq\label{eq:con2}
\bal
&P\lt( V_j^+ > t V_i^- + s m \rt) = P\lt(  V_j^+ + t V_i^+  -\f{  ti + j  }{2}   > \f{t i -j}{2} +  s m \rt) \\[1mm]
 \leq& \inf_{\th >0} \exp\lt( -\th \lt(\f{ ti - j}{2} +  s m  \rt)\rt) E\lt\{ \exp\lt(  \th \lt(  V_j^+  + t V_i^+   - \f{ j + ti } {2}  \rt)\rt) \rt\} \\[1mm]
 \leq &\inf_{\th >0} B\lt( \f{i}{m}, \f{j}{m}, t,m\th, s\rt) = \inf_{\xi >0} B\lt( \f{i}{m}, \f{j}{m}, t, \xi, s\rt) \; .
\eal
%\vspace{0.05in}
\eeq
Choosing $s = t$ in \eqref{eq:con2} establishes \eqref{eq:con}.
%Choosing $s = t$ in \eqref{eq:simp} and applying the Markov inequality to $\exp(  \th (  V_j^+  + t V_i^+   - ( j + ti) /2   ))$ for any $\th>0$, we yield
%\beq\label{eq:cdf0}
%\bal
% &P\lt(  \f{   V_j^+      }{  V_i^-  + m    }  >  t   \rt)  = P(  V_j^+ + t V_i^+  >   ti + tm  ) =  P\lt(  V_j^+ + t V_i^+  -\f{  ti + j }{2}   >  \f{ti - j}{2} + tm  \rt) \\
% & \leq \inf_{\th >0} \exp\lt( -\th \lt(\f{ ti - j}{2} + tm \rt)\rt) E \exp\lt(  \th \lt(  V_j^+  + t V_i^+   - \f{ j + ti } {2}  \rt)\rt)  \leq \inf_{\th >0} B\lt( \f{i}{m}, \f{j}{m}, t,m\th, t\rt)
%\eal
%\eeq
%The last term is exactly the upper bound in \eqref{eq:con}.
\end{proof}

One can obtain the following Hoeffding type concentration inequality using a similar argument
\vspace{0.05in}
\beq\label{eq:hoef}
P\lt( \f{V_i^+ } {V_i^- + m}  >  t \rt)  = P\lt(  V_i^+ - \f{i}{2}    >  \f{ (t-1) i  + 2tm }{2(1+t)}  \rt) \leq  \exp\lt(  -  \lt(  \f{ (t-1) i  + 2tm }{2(1+t)} \rt)^2   \f{2}{mi}    \rt)  \; ,
\vspace{0.1in}
\eeq
where $t > 1$. The key insight is that for large $i$, the above probability decays exponentially fast. %with respect to $i$ and $t$.
Thus it is possible to estimate the suprema of $V_i^+ / (V_i^- +m)$ using some covering argument.

%\begin{proof}[Proof of Theorem \ref{thm:main}]
\paragraph{Proof of Theorem \ref{thm:main} }
%Similar to \eqref{eq:cdf0}, we can use \eqref{eq:simp} to obtainand apply Markov inequality to $\exp(  \th (  V_j^+  + t V_i^+   - ( j + ti ) / 2  ))$ for any $\th >0$ and yield
%\beq\label{eq:con2}
%P\lt(  V_j^+ + t V_i^+  -\f{ j+ ti  }{2}   >  ms \rt)
%\leq  \inf_{\th >0} B\lt(\f{i}{m}, \f{j}{m}, t, m\th, s\rt)
%= \inf_{\xi > 0}B\lt(\f{i}{m}, \f{j}{m}, t, \xi, s\rt)
%\eeq
%Choosing $s = (ti-j) / (2m) + t $ recovers \eqref{eq:cdf0}.
We estimate the distribution function of the suprema. Since $V_i^+ \leq i$, we get
\beq\label{eq:cdf00}
P\lt( \sup_{s > 0} \frac{ \# \{  j \in S_0 : W_j  \geq  s \} } {      \# \{  j \in S_0 : W_j \leq - s \} + m }   >  t  \rt) = P \lt( \sup_{ i \geq 1} \f{V_i^+}{ V_i^- + m }   >  t \rt)
 =  P \lt( \sup_{ i  >  \lf tm \rf   } \f{V_i^+}{ V_i^- + m }   >  t \rt) \; .
\eeq
Once we obtain the above estimate, we can integrate $t$ from $0$ to $\infty$ to yield \eqref{pk:unif}. Based on \eqref{eq:con2} or \eqref{eq:hoef}, for any fixed $t>1$, $P(V_i^+ / (V_i^- + m) > t)$ should be exponentially small for large $i$. For small $i$, the concentration inequality is not sharp and we can use the symmetry of the joint distribution of $W_{S_0}$, i.e. $W_{S_0} \overset{d}{=} -W_{S_0}$ to obtain a better estimate. Since $V_i^+, V_i^-$ are monotone, the supremum over an interval $\sup_{ i_1 < i \leq i_2 } V_i^+ / (V_i^- + m)$ can be bounded by  $V^+_{i_2} / ( m + V^-_{i_1+1})$. We will split $i > \lf tm \rf$ into several intervals with well-chosen end points $( i_k, i_{k+1} ]$ and then apply \eqref{eq:con2}.
% it is an over estimate to use \eqref{eq:hoef} for every $j$ and then take the union bound to control the suprema.

\paragraph{Estimate for small $t < t^* =4$.}
Denote $a  = 2 \lf  tm \rf + 1$. We split the distribution into two parts:
%For $t < t^* = 4$, we have
\beq\label{eq:cdf1}
\bal
P \lt( \sup_{ i  >  \lf tm \rf   } \f{V_i^+}{ V_i^- + m }   >  t \rt) \leq P \lt( \sup_{ i  >  \lf tm \rf   } \f{V_i^+}{ V_i^- + m }   >  t , V_a^- \geq \f{a+1 }{2}  \rt) +
P \lt( V_a^-  \leq  \f{a -1}{2}  \rt)  \teq I + II  \; .
\eal
\eeq
The probability that the supremum over some small $i$ is larger than $t$ is high. Fortunately, we can use $II$ to take care of the contribution of small $i$. Since $W_{S_0} \overset{d}{=} -W_{S_0}$, we have $V_a^-  \overset{d}{=} V_a^+ = a - V_a^-$ and thus $II = 1/2$.  In $I$, $V_a^- \geq (a+1 ) / 2$ implies that the denominator is not small. Thus $V_i^+ / (V_i^- + m) > t$ cannot be true for small $i$.
%Using the monotonicity of $V_i^{\pm}$ and $V_i^+ + V_i^- = i$, we have
In fact, the monotonicity of $V_i^{\pm}$ and $V_i^+ + V_i^- = i$ imply
\[
V_a^- \geq \f{a + 1}{2}  = \lf tm\rf+ 1 \   \Rightarrow \ \sup_{ \lf tm \rf < i \leq a }  \f{V_i^+}{V_i^- + m} \leq \f{V_a^+}{ m} \leq t , \quad   \sup_{ a  < i   \leq  \lf t^2m \rf + a  }   \f{V_i^+}{V_i^- + m} \leq \f{ \lf t^2 m  \rf +  \lf  tm \rf} {  \lf tm \rf + 1 + m  } \leq t . \\
\]
Therefore, the term $I$ mainly takes care of the contribution of large $i$ and can be reduced to
\vspace{0.05in}
\[
I  = P \lt( \sup_{ i  >   \lf t^2m \rf + a   } \f{V_i^+}{ V_i^- + m }   >  t , V_a^- \geq \f{a + 1}{2}  \rt)
\leq P \lt( \sup_{ i  >   \lf t^2m \rf + a   } \f{ ( V_i^+ - V_a^+)   + \f{a + 1}{2}}{ (V_i^- - V_a^-)  + \f{a + 1}{2}+ m }   >  t \rt).
\vspace{0.05in}
\]
We freeze $V_a^{\pm}$ and introduce a new random process $U_j^{\pm} = V^{\pm}_{j+a} - V^{\pm}_a$. According to the definition of $V_j^{\pm}$ in Lemma \ref{lem:con}, we have $U_j^{\pm} = \# \{ i_{k}   : (\pm) W_{i_{k}}> 0   \ \& \  a<  k  \leq j + a   \} $ and it is the same as $V_j^{\pm}$ after throwing away $W_{i_1}, W_{i_2},..,W_{i_a}$. Thus, the random process $\{ U_j^{\pm}\}_{j \geq 1}$ and $\{ V_j^{\pm}\}_{j \geq 1}$ have the same properties and the concentration inequality \eqref{eq:con2} holds true for $U_j^{\pm}$.
%It is not difficult to verify  that the estimate \eqref{eq:mgf} holds true for the MGF of $U_j^+ + t U_i^+$.
For any increasing sequence $\{ s_i \}_{ i\geq 1}$ with $s_1 = t^2$, we obtain
\[
\bal
I & \leq P \lt( \sup_{ i  >   \lf t^2m \rf    } \f{  U_i^+    + \f{a+1}{2}}{ U_i^-  + \f{a+1}{2}+ m }   >  t \rt)
\leq  \sum_{ k \geq 1  }  P \lt( \sup_{     \lf s_k m \rf < i \leq   \lf s_{k+1} m \rf   } \f{  U_i^+    + \f{a+1}{2}}{ U_i^-  + \f{a+1}{2}+ m }   >  t \rt)    \\[1.5mm]
&\leq    \sum_{ k \geq 1  }P \lt(  \f{U_{  \lf s_{k+1} m\rf }^+ + \f{a+1}{2} }{ U_{  \lf s_{k} m\rf +1    }^- +\f{a+1}{2}+  m }  > t \rt)
= \sum_{ k \geq 1} P\lt(  U_{  \lf s_{k+1} m\rf }^+   >  t U_{  \lf s_{k} m \rf + 1 }^-  +   (t-1) \f{a+1}{2} + tm  \rt) \; .
\eal
\]
Denote $r_k =  ( t- 1) \f{a+1}{2} + tm  $. Applying \eqref{eq:con2} with $s =\f{ r_k}{m}$ gives
\[
P\lt(  U_{  \lf s_{k+1} m\rf }^+   >  t U_{  \lf s_{k} m \rf + 1 }^-  +   (t-1) \f{a+1}{2} + tm  \rt)
\leq \inf_{\xi> 0} B \lt( \f{ \lf s_{k} m \rf + 1  }{m }, \f{ \lf s_{k+1} m \rf   }{m },  t, \xi,   \f{r_k}{m} \rt) \; .
\]
%Denote $r_k =  ( t- 1) \f{a+1}{2} + tm +  \f{  t ( \lf s_{k} m \rf + 1   ) - \lf s_{k+1} m\rf    }{ 2} $. Apply \eqref{eq:con2} to $ U_{  \lf s_{k+1} m\rf }^+ + t U_{  \lf s_{k} m \rf + 1 }^+$ gives
%\[
%P\lt(  U_{  \lf s_{k+1} m\rf }^+ + t U_{  \lf s_{k} m \rf + 1 }^+ -  \f{   \lf s_{k+1} m\rf  +   t ( \lf s_{k} m \rf + 1   )   }{  2}     \geq  r_k \rt)
%\leq \inf_{\xi> 0} B \lt( \f{ \lf s_{k} m \rf + 1  }{m }, \f{ \lf s_{k+1} m \rf   }{m },  t, \xi,   \f{r_k}{m} \rt)
%\]
Recall that $a  = 2 \lf  tm \rf + 1$. We obtain
\[ \f{r_k}{m} =  (t-1) \f{a+1}{2m} + t > (t-1) t + t = t^2. \]
%\[ (t-1) \f{a+1}{2m} + t > (t-1) t + t = t^2, \quad   \f{r_k}{m} > t^2 + \f{ t s_k -s_{k+1}}{2}. \]
Using the monotonicity of $B$ in $x, y, s$ variables in   \eqref{eq:bxy}, we obtain
\[
\bal
&I \leq \sum_{k \geq 1}  \inf_{\xi> 0} B \lt( \f{ \lf s_{k} m \rf + 1  }{m }, \f{ \lf s_{k+1} m \rf   }{m },  t, \xi,  \f{r_k}{m} \rt)
\leq  \sum_{k \geq 1}  \inf_{\xi> 0} B \lt( s_k,  s_{k+1} ,  t, \xi, t^2  \rt).
\eal
\]
The upper bound of $I, II$ is independent of $m$. Thus we can estimate \eqref{eq:cdf00} uniformly for $m$.

\paragraph{Estimate for large $t \geq t^* = 4$}
%For small $t$, the upper bound is larger than $II = 1/2$.
For large $t$, $i > \lf tm \rf$ is large and \eqref{eq:con2} can be sharp. Choosing any increasing sequence $\{ s_k \}_{k \geq 1}$ with $s_1 = t$ and then applying \eqref{eq:con2} with $s= t$, we obtain
\vspace{0.05in}
\begin{align}
&P \lt( \sup_{ i  >  \lf tm \rf   } \f{V_i^+}{ V_i^- + m }   >  t \rt)  \leq
  \sum_{ k \geq 1  }  P \lt( \sup_{     \lf s_k m \rf < i \leq   \lf s_{k+1} m \rf   } \f{ V_i^+ }{  V_i^- + m }   >  t \rt)
\leq   \sum_{ k \geq 1  }  P \lt(    \f{ V_{   \lf s_{k+1} m\rf }^+ }{  V_{    \lf s_k m \rf +1      }^- + m }   >  t \rt)  \notag \\[1.5mm]
 \leq & \sum_{k\geq 1}\inf_{\xi >0} B\lt( \f{ \lf s_k m \rf +1  }{m}, \f{ \lf s_{k+1} m \rf }{m}, t, \xi,   t\rt)  \leq   \sum_{ k \geq 1}  \inf_{\xi > 0} B\lt( s_k, s_{k+1}, t,\xi,  t \rt) \; , \label{eq:cdf2}
 \vspace{0.05in}
\end{align}
where we have used the monotonicity of $B$ in $x,y,s$ variables \eqref{eq:bxy} to obtain the last inequality.
\paragraph{Choosing $s_k$} For a fixed $t$, we use a greedy strategy to optimize the selection of $s_k$ so that we have a sharp upper bound. Assume that $s_{k} , k\geq 1$ is obtained. The candidate values of $s_{k+1}$ are $ C =\{  s_k + i h :  i =1, ,2.., 19\} , \ h = (t-1) s_k / 20$. For each $s \in C$, we construct an arithmetic sequence  $ a_i \teq  (s - s_{k} ) \cdot i + s_{k} , i =0, 1,,..,30$.
Then we choose $s_{k+1}$ as follows
\beq\label{eq:nexts}
s_{k+1}  = \arg \min_{s \in C} \lt(\sum\nolimits_{i =1}^{30} \   \min_{\xi \in D_i}  B\lt( a_{i-1} , a_i, t , \xi,  \eta_t \rt) \rt),
\eeq
where $\eta_t = t^2 $ for small $t < t^*$ and $\eta_t = t$ for $t \geq t^*$.  Using $e^x + e^{-x} \leq 2 \exp(x^2/2) $, we know
 \vspace{0.05in}
\[
B(x, y, t, \xi, s) \leq \exp(-\xi ( (t x - y ) /2  +s )) \exp(\xi^2(y-x) / 8 ) \exp( (1+t)^2 \xi^2x / 8  ).
 \vspace{0.05in}
\]
%The optimal value of $\xi$ to minimize
The minimizer of the right hand side is $\xi^*(x, y, t, s) = \f{2 (tx  - y) + 4 s }{  y - x + (t + 1)^2 x  } $.  We choose
 \vspace{0.05in}
\[
D_i  = \lt\{  \xi^* /3  + 0.01 j  :   \xi^* / 3+ 0.01 j \in \lt[  \xi^* /3,  3 \xi^*  \rt]  \rt\}   , \quad \xi^* = \xi^*(  a_{i-1}, a_i, t, \eta_t  ) \; ,
 \vspace{0.05in}
\]
in \eqref{eq:nexts} and approximate $\inf_{\xi > 0}    B\lt( a_{i-1} , a_i, t , \xi, \eta_t \rt) $ by  $\min_{\xi \in D_i}    B\lt( a_{i-1} , a_i, t , \xi,   \eta_t \rt) $. We stop constructing $s_k$ if $s_k >   150$. We denote by $k_t \in Z$ the subindex of the last term and then $s_{k_t} > 150$.

\paragraph{The remaining part}
The remaining part can be arbitrary small if we construct $s_k$ over a large range and calculate large $t$ in the last step numerically. For $2.4 \leq t \leq 15$, we use the above procedure to estimate $P(\sup_{  \lf  tm \rf < i  \leq 150 m  } $ $  V_i^+ / ( V_i^- + m ) >t  )$. %To estimate the remaining part, we choose $s_i =  i , i  \geq 150$.
To estimate the remaining part $P(\sup_{ 150 m< i   } $ $  V_i^+ / ( V_i^- + m ) >t )$ , we choose $\td{s}_i \teq  s_{k_t -150+ i + 1} = i$ for $ i  \geq 150$.
From \eqref{eq:bxy}, we know
 \vspace{0.05in}
\begin{align}
B\lt( \td{s}_i, \td{s}_{i+1},t, \xi,   t   \rt)  %&B\lt( s_i, s_{i+1},t, \xi, \f{ t s_i -s_{i+1}  }{2} + t   \rt)
 =&  \exp\lt( -\xi \lt( \f{ t i - i - 1} {2} + t   \rt) \rt)   \f{  e^{\xi /2} + e^{  - \xi / 2}  }{2}       \cdot \lt( \f{   e^{ (1+t) \xi / 2 } +e^{   -(1+t) \xi / 2}  }{2}   \rt)^{ i}   \notag \\
 =& e^{-\xi t} \cdot \f{ e^{\xi} + 1 }{2} \lt( \f{ e^{\xi}   + e^{ -t \xi} }{2}    \rt)^i  = c_{\xi} \cdot e^{-t\xi} \cdot  r(t, \xi)^i \; ,
\label{eq:rem}
 \vspace{0.05in}
\end{align}
where $ r(t, \xi) \teq \f{ e^{\xi}   + e^{ -t \xi} }{2}$. We can choose $\xi \in (0,1]$ such that $ r(t,\xi) < 1 -\e$ uniformly for $t  \geq 2.4$ and some $\e > 0$. It follows that the tail $B\lt( \td{s}_i, \td{s}_{i+1},t, \xi, t   \rt) $ decays exponentially fast with respect to $i , t$. To estimate $P(\sup_{    150 m  < i  }    V_i^+ / ( V_i^- + m )> t)$, we choose $\xi = 0.2$ for $ t \in [2.4, 15]$ and obtain $r(\xi, t) <0.93$.

For $t > 15$, we choose $s_i = t + i-1 $ for $ i \geq 1$, $\xi = 0.5$ and yield $r(t, \xi) < 0.83$. Note that \eqref{eq:rem} still holds true after replacing $(\td{s}_i, \td{s}_{i+1},i)$ by $(s_i, s_{i+1}, s_i)$. Thus we can estimate the distribution function $P \lt( \sup_{ i  >  \lf tm \rf   }V_i^+ /  (V_i^- + m )   >  t \rt) $ in \eqref{eq:cdf2} directly, which decays exponentially fast with respect to $t$.
%Using \eqref{eq:rem}  by replacing $\td{s}_i$ by $s_i$, we can estimate the distribution function in \eqref{eq:cdf2} directly, which decays exponentially fast with respect to $t$.
%For $t > 15$, we choose $s_i = t + i-1 $ for $ i \geq 1$, $\xi = 0.5$ and yield $r(t, \xi) < 0.83$. Using \eqref{eq:rem} , we can estimate the distribution function in \eqref{eq:cdf2} directly, which decays exponentially fast with respect to $t$.

After obtaining the upper bound of the distribution function for $t_i = 2.4 + 0.005 i \in [2.4,15]$ and any $t > 15$ , we use the monotonicity of the distribution function and integrate \eqref{eq:cdf00} to conclude
\[
 E \lt[ \sup_{ i \geq 1} \f{V_i^+}{ V_i^- + m }    \rt]
 = \int_0^{\infty}   P \lt( \sup_{ i  >  \lf tm \rf   } \f{V_i^+}{ V_i^- + m }   >  t \rt)  dt
 \leq 2.4 + \int_{2.4}^{\infty}  P \lt( \sup_{ i  >  \lf tm \rf   } \f{V_i^+}{ V_i^- + m }   >  t \rt)  dt   \leq 3.9 \; .
\]
%\end{proof}

%\vspace{0.05in}
 \paragraph{Verification of the construction \eqref{pko:con4} of $\td{X}$.}
 Direct calculations show that
  \[
 \bal
 (X - \td{X})^T(X - \td{X}) &= [ (2 X \Sigma^{-1} - 2 \UU\CC)  \BB^{-1}]^T  [ (2 X \Sigma^{-1} - 2 \UU\CC) \BB^{-1}] \\
 &=4  \BB^{-1}  (  \Sigma^{-1} X^TX \Sigma^{-1} +  \CC^T \UU^T \UU \CC ) \BB^{-1} = 4  \BB^{-1} (  \Sigma^{-1} + \CC^T \CC )\BB^{-1} = 4 \BB^{-1} \; .\\
 (X +\td{X})^T(X - \td{X}) &=  [ X ( 2\II - 2 \Sigma^{-1} \BB^{-1})  +2  \UU \CC \BB^{-1}]^T [ 2 X \Sigma^{-1} \BB^{-1}- 2 \UU\CC \BB^{-1}] \\
&  = 4 (\II - \Sigma^{-1} \BB^{-1} )^T X^T X \Sigma^{-1} \BB^{-1} - 4 \BB^{-1} \CC^T\UU^T\UU \CC \BB^{-1} \\
&  = 4 (\II - \BB^{-1} \Sigma^{-1}) \BB^{-1} - 4 \BB^{-1} C^T C \BB^{-1} \\
& = 4 (\II - \BB^{-1} \Sigma^{-1}) \BB^{-1}  -4 \BB^{-1}(\BB - \Sigma^{-1}) \BB^{-1}  =0 \; .
 \eal
 \]
% \[
% \bal
% (X - \td{X})^T&(X - \td{X}) = [ (2 X \Sigma^{-1} - 2 \UU\CC)  \BB^{-1}]^T  [ (2 X \Sigma^{-1} - 2 \UU\CC) \BB^{-1}] \\
% =&4  \BB^{-1}  (  \Sigma^{-1} X^TX \Sigma^{-1} +  \CC^T \UU^T \UU \CC ) \BB^{-1} = 4  \BB^{-1} (  \Sigma^{-1} + \CC^T \CC )\BB^{-1} = 4 \BB^{-1} \; .\\
% (X +\td{X})^T&(X - \td{X}) =  [ X ( 2\II - 2 \Sigma^{-1} \BB^{-1})  +2  \UU \CC \BB^{-1}]^T [ 2 X \Sigma^{-1} \BB^{-1}- 2 \UU\CC \BB^{-1}] \\
%  =& 4 (\II - \Sigma^{-1} \BB^{-1} )^T X^T X \Sigma^{-1} \BB^{-1} - 4 \BB^{-1} \CC^T\UU^T\UU \CC \BB^{-1} \\
%  =& 4 (\II - \BB^{-1} \Sigma^{-1}) \BB^{-1} - 4 \BB^{-1} C^T C \BB^{-1} = 4 (\II - \BB^{-1} \Sigma^{-1}) \BB^{-1}  -4 \BB^{-1}(\BB - \Sigma^{-1}) \BB^{-1}  =0 \; .
% \eal
% \]
 Here we use $\UU^T X = X^T \UU = 0$. The first identity implies \eqref{pko:con0} and the second is exactly \eqref{pko:ort}.

\section{Proof of  Lemma \ref{lem:cov}}
% \paragraph{Proof of  Lemma \ref{lem:cov}}
%\begin{proof}
 Conditional on $\eta$, we can determine $N_{\eta} =\{ j \in S_0: W_j \ne 0\}$. Recall that $\xi$ and $\eta$ are independent and $\xi_{S_0} \overset{d}{=} -\xi_{S_0}$. We have
 $E(\one_{\xi_i > 0} | \eta) =E(\one_{\xi_i > 0}) =  1/2, i\in S_o$. For any $i, j \in N_{\eta}$, we get
 \[
    E(\one_{\xi_i > 0} \one_{\xi_j<0})  - \frac{1}{4} +    E(\one_{\xi_i > 0} \one_{\xi_j>0})  -\frac{1}{4}  =  E(\one_{\xi_i > 0}) -\frac{1}{2}=0 \; .
 \]
 Similarly, we have $E(\one_{\xi_i < 0} \one_{\xi_j > 0})  - \frac{1}{4}= -( E(\one_{\xi_i > 0} \one_{\xi_j>0})  -\frac{1}{4} ) , \forall i, j\in S_0$. Meanwhile, the symmetry of $\xi_{S_0}$ implies $ E(\one_{\xi_i < 0} \one_{\xi_j < 0}) =  E(\one_{\xi_i > 0} \one_{\xi_j>0})$. Therefore, we obtain
% \vspace{-0.05in}
  \beq\label{cond:fdr1}
  \bal
 & \Cov(Y_i ,Y_j | \eta) = E(Y_i Y_j | \eta )  - E(Y_i |\eta) E(Y_j |\eta) = E(Y_i Y_j | \eta)  - \frac{1}{4} \\
  =& E[    (  \one_{f_i(\eta) >0} \one_{\xi_i >0} +  \one_{f_i(\eta) <0} \one_{\xi_i <0}    ) \cdot (       \one_{f_j(\eta) >0} \one_{\xi_j >0} +  \one_{f_j(\eta) <0} \one_{\xi_j<0}           )        | \eta     ]-\frac{1}{4} \\[1mm]
  =& \one_{f_i(\eta) >0}   \one_{f_j(\eta) >0}  \left[ E(   \one_{\xi_i >0}   \one_{\xi_j >0}      ) - \frac{1}{4} \right] +  \one_{f_i(\eta) >0}   \one_{f_j(\eta) <0}  \left[ E(   \one_{\xi_i >0}   \one_{\xi_j <0}      ) - \frac{1}{4}  \right] \\[1mm]
   &\quad  +  \one_{f_i(\eta) <0}   \one_{f_j(\eta) >0}  \left[ E(   \one_{\xi_i <0}   \one_{\xi_j >0}      ) - \frac{1}{4}  \right] +  \one_{f_i(\eta) <0}   \one_{f_j(\eta) <0}  \left[ E(   \one_{\xi_i <0}   \one_{\xi_j <0}      ) - \frac{1}{4}  \right] \\[1mm]
   =& ( E(\one_{\xi_i > 0} \one_{\xi_j>0})  -\frac{1}{4}   ) (  \one_{f_i(\eta) >0, f_j(\eta) >0  } -    \one_{f_i(\eta) >0,f_j(\eta) <0 }   -\one_{f_i(\eta) <0 , f_j(\eta) >0}+          \one_{f_i(\eta) <0 ,f_j(\eta) <0}    )   \\[1mm]
   =& ( E(\one_{\xi_i > 0} \one_{\xi_j>0})  -\frac{1}{4}   )  (  \one_{f_i(\eta) >0} -    \one_{f_i(\eta) <0}  )(  \one_{f_j(\eta) >0} -    \one_{f_j(\eta) <0}  )
 = ( E(\one_{\xi_i > 0} \one_{\xi_j>0})  -\frac{1}{4}   ) w_iw_j,
 \eal
 \vspace{0.05in}
 \eeq
 where $ w_i \triangleq  \one_{f_i(\eta) >0} -    \one_{f_i(\eta) <0}  $. By definition, $w_i = 1$ or $-1$. From $\Cov(\xi) = \BB = 2\Sigma^{-1}$, we know
 that
 $
 \left(\begin{array}{c}
 \xi_i \\
 \xi_j \\
 \end{array}
 \right)
 \sim
 N \left(
 0,
 \left(
 \begin{array}{cc}
 \BB_{ii} & \BB_{ij} \\
 \BB_{ji} & \BB_{jj} \\
 \end{array}
 \right)
 \right) . $
 Since normalizing $\xi_i,\xi_j$ does not change their sign, we  assume that
 $\left(
 \begin{array}{c}
 \xi_i \\
 \xi_j \\
 \end{array}
 \right)
 \sim
 N \left(
 0,
 \left(
 \begin{array}{cc}
 1 & \mu_{ij} \\
 \mu_{ij} & 1 \\
 \end{array}
 \right)
 \right) , $
 where $ \mu_{ij} = \BB_{ij} / ( \BB^{1/2}_{ii}  \BB^{1/2}_{jj}) = (\wt{\Sigma^{-1}})_{ij}$ (see \eqref{def:nota2}).
 Define $ \mu = \mu_{ij} =( \wt{\Sigma^{-1}})_{ij}$
% \begin{equation}
% \mu = \mu_{ij} =( \wt{\Sigma^{-1}})_{ij},
% \label{def-mu}
% \end{equation}
 and let $P(\xi_i, \xi_j)$ and $P_s(\cdot)$ be the probability distribution function of $(\xi_i ,\xi_j)$ and the standard normal distribution, respectively.
 %Using $0 \leq e^x - 1 -x \leq \frac{x^2}{2} (e^x  1_{x>0} + 1)$ by taking  $x \triangleq   -\frac{\mu^2\xi_i^2 + \mu^2\xi_j^2 - 2\mu \xi_i \xi_j}{2(1-\mu^2)}$,
 Using
 \[
 0 \leq e^x - 1 -x \leq \frac{x^2}{2} (e^x  1_{x>0} + 1) , \quad x \triangleq   -\frac{\mu^2\xi_i^2 + \mu^2\xi_j^2 - 2\mu \xi_i \xi_j}{2(1-\mu^2)} ,
 \]
we expand $P(\xi_i, \xi_j) -  P_s(\xi_i) P_s(\xi_j) $ up to $\mu^2$
 \[
 \bal
 \left[ P(\xi_i, \xi_j) -  P_s(\xi_i) P_s(\xi_j)\right]  w_i w_j =& \frac{  P_s(\xi_i) P_s(\xi_j)}{\sqrt{1-\mu^2}}
 (  1+ x -\sqrt{1-\mu^2}  + e^x -1 -x )w_i w_j  \\[1mm]
 \leq & \frac{  P_s(\xi_i) P_s(\xi_j)}{\sqrt{1-\mu^2}}(   ( 1+ x -\sqrt{1-\mu^2} )w_i w_j  + e^x -1 -x )   \\[1mm]
  \leq & \frac{  P_s(\xi_i) P_s(\xi_j)}{\sqrt{1-\mu^2}} \lt(   ( 1+ x -\sqrt{1-\mu^2} )w_i w_j  +\frac{x^2 (e^x \one_{x>0} + 1)  }{2}\rt)  \\[1mm]
  =&  \frac{  P_s(\xi_i) P_s(\xi_j)}{\sqrt{1-\mu^2}} \left(  ( 1+ x -\sqrt{1-\mu^2} )w_i w_j+  \frac{x^2 }{2} \right)+ P(\xi_i, \xi_j) \frac{x^2  \one_{x>0} }{2} \; .
  \eal
 \]
 Integrating both sides with respect to $\xi_i,\xi_j$ in the region $\xi_i,\xi_j >0$ gives
 %\beq\label{pko:cov2_0}
 \begin{align}
 \left(E(\one_{\xi_i>0}\one_{\xi_j >0})-\frac{1}{4}\right) w_i w_j& \leq \int_{\xi_i, \xi_j>0} \frac{  P_s(\xi_i) P_s(\xi_j)}{\sqrt{1-\mu^2}} \left(  ( 1+ x -\sqrt{1-\mu^2} )w_i w_j+  \frac{x^2 }{2} \right)d\xi_i d\xi_j \notag \\[1.2mm]
 &+ \int_{\xi_i>0, \xi_j>0}  P(\xi_i, \xi_j) \frac{x^2  \one_{x>0} }{2} d\xi_i d\xi_j  \triangleq \II + \II\II + \II\II\II \; .\label{pko:cov2_0}
 \end{align}
 Since $P_s(\cdot)$ is a standard Gaussian distribution and $x =   -\frac{\mu^2\xi_i^2 + \mu^2\xi_j^2 - 2\mu \xi_i \xi_j}{2(1-\mu^2)} $, we can calculate all the moments in $\II, \II\II$ explicitly. For $\II$, we have
 \beq\label{pko:cov2_1}
 \bal
 \II &= \left( \frac{1}{4} (\frac{1}{\sqrt{1-\mu^2} }-1)  + \frac{\mu}{2\pi ( 1-\mu^2)^{3/2}  } - \frac{1}{4} \frac{\mu^2}{  ( 1-\mu^2)^{3/2 } }  \right) w_i w_j\\
 &\leq  \frac{\mu w_i w_j} {2\pi} + \Big| \frac{\mu}{2\pi} (  (1-\mu^2)^{-3/2} -1  )  \Big| + \Big|    \frac{1}{4} (\frac{1}{\sqrt{1-\mu^2} }-1)    - \frac{1}{4} \frac{\mu^2}{  ( 1-\mu^2)^{3/2 } }   \Big|  \\ %\triangleq \II_1 + \II_2 + \II_3 \\
  &= \frac{\mu w_i w_j} {2\pi} + \frac{\mu^2}{2\pi} \Big| \frac{    (1-\mu^2)^{-3/2} -1     }{   \mu}\Big|
  +  \frac{\mu^2}{4\sqrt{1-\mu^2} } \Big|  \frac{1}{1-\mu^2} - \frac{1}{1+\sqrt{1-\mu^2}} \Big| \triangleq  \frac{\mu w_i w_j} {2\pi} + c_1(\mu) \mu^2, \\
  \eal
 \eeq
 where $c_1(\mu) \geq 0$ collects the coefficients of $\mu^2$ and is bounded near $\mu =0$. We use $E(\xi \one_{\xi >0}) = 1/\sqrt{2\pi}, E(\xi^2 \one_{\xi >0})  = \frac{1}{2}$ for the standard Gaussian $\xi$ to obtain the first equality, and $|w_i| =|w_j| =1$ to obtain the inequality. %Note that $x$ is the order $\mu$.
 For the second term, we get
 \beq\label{pko:cov2_2}
 \II\II =\frac{\mu^2}{8(1-\mu^2)^{5/2}} \int_{\xi_i, \xi_j>0} P_s(\xi_i) P_s(\xi_j) (2\xi_i \xi_j -\mu(\xi_i^2 +\xi_j^2))^2 d\xi_i d\xi_j = \frac{\mu^2(1  - \frac{8}{\pi} \mu + 2 \mu^2) }{8(1-\mu^2)^{5/2}} \triangleq c_2(\mu) \mu^2,
 \eeq
 where $c_2(\mu) =  \frac{(1  - \frac{8}{\pi} \mu + 2 \mu^2) }{8(1-\mu^2)^{5/2}}  \geq 0$ is bounded near $\mu =0$. Since $\xi_i,\xi_j >0$ and $x =  -\frac{\mu^2\xi_i^2 + \mu^2\xi_j^2 - 2\mu \xi_i \xi_j}{2(1-\mu^2)}$,  $\mu \leq 0$ implies $x \leq 0$, or equivalently $\one_{x > 0} \leq \one_{\mu > 0}$. Note that
 \[
 x =  -\frac{\mu^2\xi_i^2 + \mu^2\xi_j^2 - 2\mu \xi_i \xi_j}{2(1-\mu^2)} \leq  -\frac{ 2\mu^2\xi_i\xi_j  - 2| \mu| \xi_i \xi_j}{2(1-\mu^2)} =   \frac{ |\mu|  \xi_i \xi_j}{(1 + |\mu| )}, \ \forall \  \xi_i, \xi_j >0.
  \]
 For $\xi_i ,\xi_j > 0$, we have
  $ x^2 \one_{x >0} \leq  \left( \frac{ | \mu| \xi_i \xi_j  }{1+ |\mu|}\right)^2 \one_{\mu >0} .$
 Therefore, we obtain
 \beq\label{pko:cov2_3}
 \bal
 \II \II \II &= \frac{  1  }{2}  E(x^2 \one_{x>0} \one_{\xi_i ,\xi_j >0} ) \leq  \frac{1}{2} \frac{|\mu^2 \one_{\mu >0} |}{ (1+|\mu|)^2} E(\xi_i^2 \xi_j^2 \one_{\xi_i,\xi_j >0}) \\[1mm]
 &\leq  \frac{1}{2} \frac{|\mu^2 \one_{\mu >0} |}{ (1+|\mu|)^2} ( E(\xi_i^4 \one_{\xi_i >0})E(\xi_j^4 \one_{\xi_j >0}) )^{1/2} =  \frac{1}{2} \frac{|\mu^2 \one_{\mu >0} |}{ (1+|\mu|)^2} \frac{3}{2} \triangleq \one_{\mu>0}c_3(\mu) \mu^2,
 \eal
 \eeq
 where $c_3(\mu) = \frac{3}{4(1+ |\mu|)^2 }$ is bounded near $\mu =0$. Combining \eqref{pko:cov2_0}, \eqref{pko:cov2_1}, \eqref{pko:cov2_2} and \eqref{pko:cov2_3}, we yield
 \[
 \Cov(Y_i, Y_j| \eta) = \left[E(\one_{\xi_i>0}\one_{\xi_j >0})-\frac{1}{4}\right] w_i w_j \leq \frac{\mu}{2\pi} w_i w_j + (c_1(\mu) + c_2(\mu)  + c_3(\mu) \one_{\mu >0}) \mu^2 \triangleq  \frac{\mu}{2\pi} w_i w_j + c(\mu) \mu^2.
 \]
 Here, $c(\mu) =c_1(\mu) + c_2(\mu)  + c_3(\mu) \one_{\mu >0} $. Since $c_i(\mu)$ is a non-negative and an explicit function of $\mu$, it is not difficult to show that $c(\mu) < \frac{3}{2}$ for $|\mu | < \frac{1}{2}$. For $| \mu| > 1/2$, we use the estimate
 $
 \Cov(Y_i, Y_j| \eta) \leq 1/4 \leq - \f{\mu}{2\pi} + \frac{3}{2} \mu^2.
 $
 Finally, we conclude
 \[
 \Cov(Y_i, Y_j| \eta) \leq  \left(\frac{\mu}{2\pi} w_i w_j + c(\mu) \mu^2 \right)\wedge \frac{1}{4} \leq    \frac{\mu}{2\pi} w_i w_j  + \frac{3}{2}\mu^2 \; ,
 \]
 where $\mu = ( \wt{\Sigma^{-1}})_{ij}$. This proves Lemma \ref{lem:cov}.

 \vspace{0.2in}
  \noindent
  {\bf Acknowledgements.}
% The research was in part supported by NSF Grants DMS 1318377 and DMS 1613861. The research of the first author was performed during his visit to ACM at Caltech.  We would like to thank Professor Emmanuel Candes for his many valuable comments and suggestions to our work. We would also like to thank Professor Lucas Janson for his interest and comments on the earlier version of this manuscript and Dr. Pengfei Liu for the insightful discussions on the pseudo knockoff.
 The research was in part supported by NSF Grants DMS 1318377 and DMS 1613861. The research of JC was performed during his visit to ACM at Caltech.  We would like to thank Professor Emmanuel Candes for his many valuable comments and suggestions to our work. We would also like to thank Professor Lucas Janson for his interest and comments on the earlier version of this manuscript and Dr. Pengfei Liu for the discussions on the pseudo knockoff.

 \end{document}